\documentclass[singlecolumn, journal]{IEEEtran}
\usepackage[cmex10]{amsmath}
\usepackage{graphicx} 
\usepackage[numbers,sort&compress]{natbib}
\usepackage{epsfig}
\usepackage{amsmath}
\usepackage{amssymb}
\usepackage{footnote}
\usepackage{array}
\usepackage{verbatim, animate}
\usepackage{comment, amsmath, amsfonts, amsthm}
\usepackage[mathscr]{euscript}
\usepackage{epsfig}
\usepackage{amsfonts,amssymb,latexsym,amsmath, amsxtra,verbatim}
\usepackage[all]{xy}
\usepackage{psfrag}
\usepackage{xcolor}

\usepackage{xspace}
\usepackage{bbm}
\input{mathlig}

\newcommand{\muspace}{\mspace{1mu}}

\DeclareRobustCommand{\scond}{\mathchoice{\muspace\vert\muspace}{\vert}{\vert}{\vert}}
\mathlig{|}{\scond}

\DeclareRobustCommand{\discint}{\mathchoice{\mspace{-1.5mu}:\mspace{-1.5mu}}{\mspace{-1.5mu}:\mspace{-1.5mu}}{:}{:}}
\mathlig{::}{\discint}

\def\var{\mathop{\rm Var}\nolimits}%
\def\co{\mathop{\rm co}\nolimits}%

\newcommand{\Pc}{\mathcal{P}}

\newcommand{\Sc}{\mathcal{S}}

\newcommand{\Xc}{\mathcal{X}}
\newcommand{\Yc}{\mathcal{Y}}

\newcommand{\Ncal}{\mathcal{N}}

\newcommand{\Scal}{\mathcal{S}}

\newcommand{\Ucal}{\mathcal{U}}

\newcommand{\Xcal}{\mathcal{X}}
\newcommand{\Ycal}{\mathcal{Y}}

\newcommand{\Cr}{\mathscr{C}}

\newcommand{\Rr}{\mathscr{R}}

\newcommand{\Yt}{{\tilde{Y}}}

\newcommand{\yt}{{\tilde{y}}}

\def\a{\alpha}

\def\e{\epsilon}

\def\l{\lambda}

\DeclareMathOperator\E{\textsf{E}}

\newcommand{\Bern}{\mathrm{Bern}}

\newcommand{\U}{\mathrm{Unif}}

\def\textiid{i.i.d.\@\xspace}
\newcommand\iid{\ifmmode\text{ i.i.d. } \else \textiid \fi}

\def\mathllap{\mathpalette\mathllapinternal}
\def\mathllapinternal#1#2{%
  \llap{$\mathsurround=0pt#1{#2}$}}

\def\clap#1{\hbox to 0pt{\hss#1\hss}}
\def\mathclap{\mathpalette\mathclapinternal}
\def\mathclapinternal#1#2{%
  \clap{$\mathsurround=0pt#1{#2}$}}

\let\oldstackrel\stackrel
\renewcommand{\stackrel}[2]{\oldstackrel{\mathclap{#1}}{#2}}

\renewcommand{\hbar}{h\mathllap{\overline{\vphantom{h}\hphantom{\rule{4.6pt}{0pt}}}\mspace{0.77mu}}}

\catcode`~=11 
\newcommand{\urltilde}{\kern -.06em\lower -.06em\hbox{~}\kern .02em}
\catcode`~=13 

\usepackage{lineno}
\usepackage{auto-pst-pdf}
\usepackage{pst-pdf}
\usepackage[font=footnotesize]{caption}
\usepackage{subcaption}
\usepackage{multirow}

\newtheorem{theorem}{\textbf{Theorem}}

\newtheorem{lemma}{\textbf{Lemma}}

\newtheorem{proposition}{\textbf{Proposition}}
\newtheorem{definition}{\textbf{Definition}}
\newtheorem{example}{\textbf{Example}}
\newtheorem{remark}{\textbf{Remark}}
\newcommand{\bq}{\bar{p}_2}
\newcommand{\bp}{\bar{p}_1}
\newcommand{\q}{p_2}
\newcommand{\p}{p_1}
\newcommand{\Ytcal}{\tilde{\Ycal}}
\newcommand{\argmax}{\operatornamewithlimits{arg\,max}}
\newcommand{\CC}{\mathfrak{C}}

\onecolumn
\IEEEoverridecommandlockouts
\makeatother
\begin{document}
\title{Capacity Theorems for Broadcast Channels with Two Channel State Components Known at the Receivers} 

\author{Hyeji Kim and Abbas El Gamal\\
Department of Electrical Engineering\\
 Stanford University\\
 Email: hyejikim@stanford.edu, abbas@ee.stanford.edu
\thanks{ This work was partially supported by Air Force grant FA9550-10-1-0124.  This paper was presented in part at \emph{Proc. IEEE Int. Symp. Inf. Theory, Hawaii, 2014}.  }%
}
\maketitle

\begin{abstract} 
We establish the capacity region of several classes of broadcast channels with random state in which the channel to each user is selected from two possible channel state components and the state is known only at the receivers. When the channel components are deterministic, we show that the capacity region is achieved via Marton coding. 
This channel model does not belong to any class of broadcast channels for which the capacity region was previously known and is useful in studying wireless communication channels when the fading state is known only at the receivers. We then establish the capacity region when the channel components are ordered, e.g., degraded. In particular we show that the capacity region for the broadcast channel with degraded Gaussian vector channel components is attained via Gaussian input distribution. Finally, we extend the results on ordered channels to two broadcast channel examples with more than two channel components, but show that these extensions do not hold in general.
\end{abstract}
\section{Introduction}
Consider the discrete memoryless broadcast channel (DM-BC) with random (IID) state $(\Xc \times \Sc, p(y_1,y_2|x,s)p(s), \Yc_1 \times \Yc_2)$ with the state $S$ known only at the receivers. Assume the setup in which the sender wishes to transmit a common message $M_0 \in [1:2^{nR_0}]$ to both receivers and private messages $M_j \in [1:2^{nR_j}]$ to receiver $j\in \{1,2\}$ as depicted in Figure~\ref{figure1}.  

It is well known that this broadcast channel with state setup can be viewed as a general DM-BC with input $X$ and outputs $(Y_1,S)$ and $(Y_2,S)$. Hence the definitions of a $(2^{nR_0},2^{nR_1},2^{nR_2})$ code, achievability and the capacity region $\Cr$ are the same as for the general broadcast channels~\cite{EG--Kim2011}. Moreover, the capacity region for this broadcast channel with state setup is not known in general. The Marton inner bound and the \emph{UV} outer bound on the general broadcast channel hold for this channel and they coincide when the channel $X \to (Y_1,S),(Y_2,S)$ falls into any of the classes of the broadcast channel for which the capacity region is known (see \cite{Geng--Gohari--Nair--Yu2014} for examples of these classes). Beyond these classes, there have been some efforts on evaluating inner bounds on the capacity region of the Gaussian fading BC model, including superposition coding by Jafarian and Vishwanath~\cite{Jafarian--Vishwanath2011}, time division with power control by Liang and Goldsmith~\cite{Liang--Goldsmith2005}, and superposition of binary inputs motivated by a capacity achieving strategy for a layered erasure broadcast channel by Tse and Yates~\cite{Tse--Yates2012}. 

\begin{figure}[h]
\begin{center}
\psfrag{M}[r]{$M_0, M_1, M_2\ $}
\psfrag{E}[c]{Encoder}
\psfrag{X}[c]{$X^n$}
\psfrag{S}[c]{$p(s)$}
\psfrag{T}[c]{$S^n$}
\psfrag{G}[c]{$p(y_1,y_2|x,s)$}
\psfrag{Y}[c]{$Y_1^n$}
\psfrag{C}[c]{$Y_2^n$}
\psfrag{D}[c]{Decoder 1}
\psfrag{F}[c]{Decoder 2}
\psfrag{A}[l]{$\hat{M}_{01}, \hat{M}_1$}
\psfrag{B}[l]{$\hat{M}_{02}, \hat{M}_2$}
\includegraphics[scale=0.4]{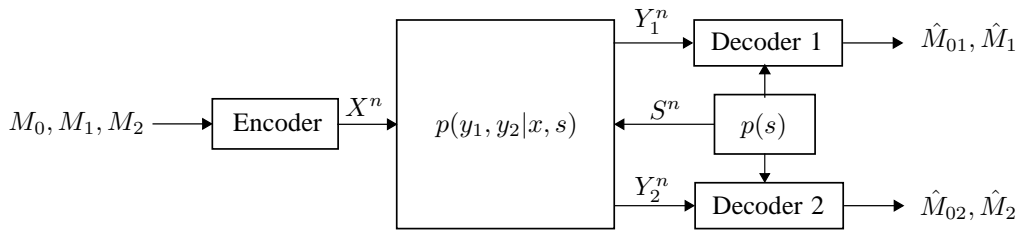}\\
\end{center}
\caption{Broadcast channel with state known only at the receivers.}\label{figure1}
\end{figure}


In this paper we focus on the special class of the broadcast channel with state in Figure~\ref{figure1} in which there are only two channel state components as depicted in Figure~\ref{fig:model}. In this model, which we refer to as the {\em broadcast channel with two channel states} (BC-TCS), the state $S=(S_1,S_2)\in \{1,2\}^2$ with $p_{S_1}(1)=p_1, p_{S_1}(2)=1-p_1=\bar p_1$ and $p_{S_2}(1)=p_2, p_{S_2}(2)=\bar p_2$, and the two possible channel components are denoted by $\Yt_1 \sim p(\tilde{y}_1|x)$ and $\Yt_2 \sim p(\tilde{y}_2|x)$. The outputs of the BC-TCS is 
\allowdisplaybreaks
\begin{align}\begin{split}\label{model}
Y_1 = \begin{cases}	\Yt_1 &\text{ if } S_1=1,\\
				\Yt_2 &\text{ if } S_1=2,
	\end{cases}\\
Y_2 = \begin{cases} \Yt_1 &\text{ if } S_2=1,\\
				\Yt_2 &\text{ if } S_2= 2.
	\end{cases}
\end{split} \end{align}
Without loss of generality, we assume throughout that $p_1 \geq p_2$ and that receiver $j=1,2$ knows the state sequence $S^n$ but the sender does not.~\footnote{Since the capacity region of the broadcast channel depends only on its marginal distributions~\cite{EG--Kim2011}, we only need to specify the marginal pmfs of $S_1$ and $S_2$. Moreover, it suffices to assume that receiver $j=1,2$ knows only its state sequence $S_j^n$.}

\begin{figure}[h]
\begin{center}
\psfrag{X}[l]{$X$}
\psfrag{F}[c]{\textcolor{black}{$~p_{\Yt_1|X}(y_1|x)$}}
\psfrag{G}[c]{\textcolor{black}{$~p_{\Yt_2|X}(y_2|x)$}}
\psfrag{H}[c]{\textcolor{black}{$~p_{\Yt_1|X}(y_2|x)$}}
\psfrag{I}[c]{\textcolor{black}{$~p_{\Yt_2|X}(y_1|x)$}}
\psfrag{Y}[l]{$Y_1$}
\psfrag{Z}[l]{$Y_2$}
\small
\psfrag{A}[c]{(a) $(S_1,S_2)=({\color{black}1},{\color{black}1}), p_{S_1}(1)=p_1, p_{S_2}(1)=p_2$}
\psfrag{B}[c]{(b) $(S_1,S_2)=({\color{black}1},{\color{black}2}), p_{S_1}(1)=p_1, p_{S_2}(2)=\bar{p}_2$}
\psfrag{C}[c]{(c) $(S_1,S_2)=({\color{black}2},{\color{black}1}), p_{S_1}(2)=\bar{p}_1, p_{S_2}(1)=p_2$}
\psfrag{D}[c]{(d) $(S_1,S_2)=({\color{black}2},{\color{black}2}), p_{S_1}(2)=\bar{p}_1, p_{S_2}(2)=\bar{p}_2$}
\includegraphics[scale=0.7]{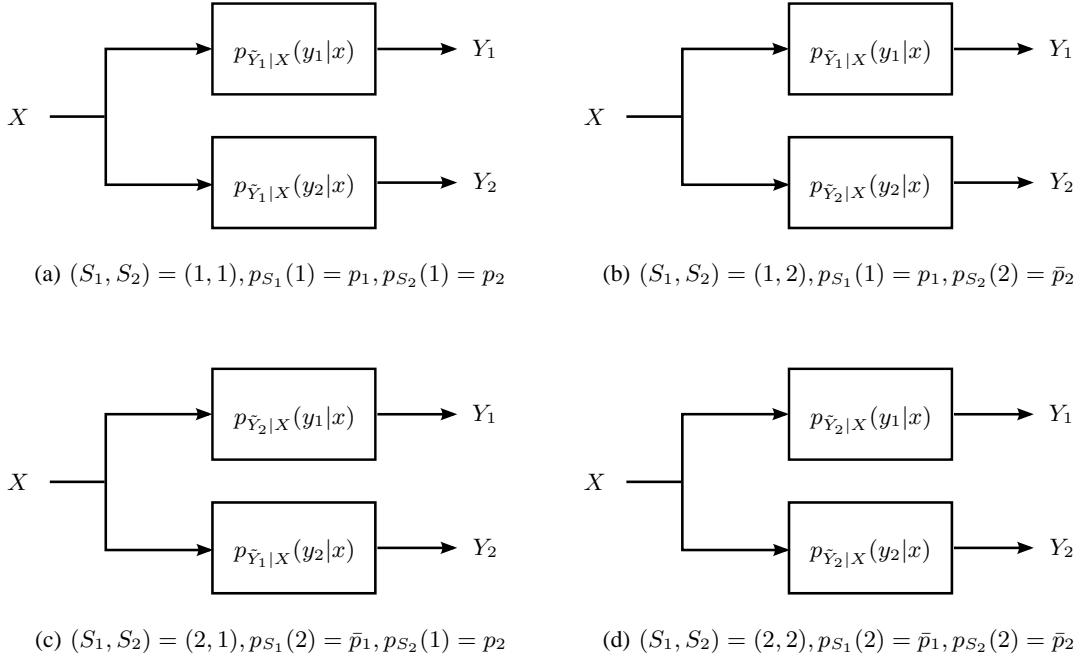}
\end{center}
\caption{Broadcast channel with two channel state components.}\label{fig:model}
\end{figure}


In~\cite{Kim--EG2014}, we established the private message capacity region for the special case of the BC-TCS in which the state components are deterministic functions, i.e., $\Yt_1 = f_1(X)$ and $\Yt_2 = f_2(X)$. Achievability is established using Marton coding~\cite{Marton1979}. The key observation is that the auxiliary random variables in the Marton region characterization, $U_1$ and $U_2$, are always set to $f_1,f_2,X$, or $\emptyset$. In particular if the channel $p(y_1|x)$  is more likely to be $f_1$ than the channel $p(y_2|x)$, then $(U_1, U_2)$ are set to $(X,\emptyset),(\emptyset,X)$, or $(f_1,f_2)$. The converse is established by showing that the Marton inner bound with these extreme choices of auxiliary random variables coincides with the \emph{UV} outer bound~\cite{Nair--EG2007}. It is important to note that this class of broadcast channels with two deterministic channel components (BC-TDCS) does not belong to any class of broadcast channels for which the capacity region is known. It also provides yet another class of broadcast channels for which Marton coding is optimal. Moreover, the BC-TDCS model can be used to approximate certain fading broadcast channels in high SNR (see Example~\ref{ex:finite} in Section~\ref{sec:Deterministic}). 

In this paper we provide a complete proof for the result in~\cite{Kim--EG2014} and extend it to the case with common message (see Section~\ref{sec:Deterministic}). In addition, we include several new results on the capacity region of the BC-TCS. In Section~\ref{sec:Ordered}, we study the case when the channel components are ordered, which models, for example, a wireless downlink channel in which the channel to each user can be either ``strong'' or ``weak''. We show that if the BC $p(\yt_1,\yt_2|x)$ is degraded, less noisy, more capable, or dominantly c-symmetric, then the corresponding BC-TCS $p(y_1,y_2,s|x)$ is degraded, less noisy, more capable, or dominantly c-symmetric, and the capacity region is achieved via superposition coding. This is surprising (and as we will show does not extend to more than two components in general) because the sender does not know the state, hence does not know which of the two channels $p(y_1,s|x)$ or $p(y_2,s|x)$ is stronger. We further show that the capacity region of the BC-TCS with degraded Gaussian vector channel components, which is a special case of the BC-TCS with degraded channel components, is attained by Gaussian channel input. This is again unexpected because for the general degraded fading Gaussian BC (where we know that one channel is always a degraded version of the other), the optimizing input distribution is not Gaussian~\cite{Abbe--Zheng2009}. 
In Section~\ref{sec:More}, we present results on the broadcast channel with more than two channel components. We establish the capacity region when there are three BEC or BSC channel components and show that there is a gap between the Marton inner bound and the UV outer bound when there are four BSC channel components. Hence our results for the two channel state components do not extend to more than two state components in general.

\section{Deterministic channel state components}\label{sec:Deterministic}
In this section, we consider the BC-TCS with two deterministic channel components $\Yt_1 = f_1(X)$ and 
$\Yt_2 = f_2(X)$, henceforth referred to as BC-TDCS. We show that the capacity region of the BC-TDCS is achieved using Marton coding.

\begin{theorem}[private message capacity region for BC-TDCS~\cite{Kim--EG2014}]\label{cap:private}
\textnormal{
The private message capacity region of the BC-TDCS $(\Xc \times \Sc, p(s)p(y_1,y_2|x,s), \Yc_1 \times \Yc_2)$ with the state known only at the receivers is
\allowdisplaybreaks
\begin{align}
\begin{split}\label{capacity:BC-TDCS}
\Cr &= \co\{\Rr_1 \cup \{(C_1,0)\} \cup \{(0,C_2)\}\},
\end{split}\end{align}
where $C_j = \max_{p(x)} I(X;Y_j|S)$ for $j=1,2$, and
\begin{align*}
 \Rr_1 = \{(R_1,R_2)\colon R_{1} &\leq I(f_1;Y_1|S),\\
					     R_{2} &\leq I(f_2;Y_2|S),\\
					     R_1+R_2 &\leq I(f_1;Y_1|S)+I(f_2;Y_2|S)-I(f_1;f_2) \text{ for some }p(x) \in \Pc\},
\end{align*}
where
\begin{align*}
\Pc = \{&\argmax_{p(x)}\ \big(I(f_1;Y_1|S)+\lambda (I(f_2;Y_2|S)-I(f_1;f_2))\big) \text{ for }\bar{p}_1/\bar{p}_2 \leq \lambda \leq 1,\\
&\argmax_{p(x)}\ (I(f_1;Y_1|S)-I(f_1;f_2)+\nu I(f_2;Y_2|S)) \text{ for } 1 \leq \nu \leq p_1/p_2\}.
\end{align*}
}\end{theorem}

\begin{proof}
For achievability we use Marton coding which achieves the set of rate pairs $(R_1,R_2)$ such that
\begin{align}\label{capacity:BC-TDCS}
\begin{split}
R_1 &\leq I(U_1;Y_1|S),\\
R_2 &\leq I(U_2;Y_2|S),\\
R_1+R_2 &\leq I(U_1;Y_1|S)+I(U_2;Y_2|S)-I(U_1;U_2)
\end{split}
\end{align}
for some pmf $p(u_1,u_2,x)$.

Note that the rate pair $(C_1,0)$ satisfies the inequalities~\eqref{capacity:BC-TDCS} for $p(x) = \arg\max I(X;Y_1|S)$ and $(U_1,U_2) = (X,\emptyset)$. Similarly, the rate pair $(0,C_2)$ satisfies the inequalities~\eqref{capacity:BC-TDCS} for $p(x)= \arg\max I(X;Y_2|S)$ and $(U_1,U_2) = (\emptyset, X)$. Thus $(C_1,0)$ and $(0,C_2)$ are achievable. Now let $\Rr'_1$ be the set of rate pairs that satisfy~\eqref{capacity:BC-TDCS} for some $p(x)$ and $(U_1,U_2)=(f_1,f_2)$. We can easily see that $\Rr_1 \subseteq \Rr'_1$. Thus $\Cr$ is achievable via Marton coding and time-sharing.


To establish the converse, we show that $\Cr$ coincides with the \emph{UV} outer bound. The \emph{UV} outer bound for the broadcast channel with state known at the receivers states that if a rate pair $(R_1,R_2)$ is achievable, then it must satisfy the inequalities
\begin{align}\begin{split}\label{uv}
R_1 &\leq I(U;Y_1|S),\\
R_2 &\leq I(V;Y_2|S),\\
R_1+R_2 &\leq I(U;Y_1|S)+I(X;Y_2|U,S),\\
R_1+R_2 &\leq I(V;Y_2|S)+I(X;Y_1|V,S)
\end{split}\end{align}
for some pmf $p(u,v,x)$.
Let this outer bound be denoted by $\bar \Rr$. Clearly $\Cr \subseteq \bar \Rr$. We now show that every supporting hyperplane of $\bar{\Rr}$ intersects $\Cr$, i.e., 
for all $\lambda \geq 0$,
\begin{align}\label{ineq1} 
\max_{(R_1,R_2)\in\bar{\Rr}} (R_1+\lambda R_2) \leq \max_{(r_1,r_2)\in\Cr} (r_1+\lambda r_2).
\end{align}
We first show that inequality~\eqref{ineq1} holds for $0\leq \lambda \leq 1$. Consider
\begin{align}
\max_{(R_1,R_2) \in \bar \Rr}(R_1+\lambda R_2)&\le \max_{p(u,x)}(I(U;Y_1|S)+\lambda H(Y_2|U,S))\nonumber\\
&=\max_{p(x)}\Big(H(Y_1|S)+\max_{p(u|x)}( \lambda H(Y_2|U,S)-H(Y_1|U,S))\Big)\nonumber\\
&=\max_{p(x)}\Big(\p H(f_1)+ \bp H(f_2)  +\max_{p(u|x)}( (\lambda \bq-\bp)H(f_2|U) + (\lambda \q-\p)H(f_1|U)) \Big).\nonumber
\end{align}
We now consider different ranges of $0\leq \lambda \le 1$. 
\begin{itemize}
\item For $0\leq \lambda \leq \bp/\bq$, $(\lambda\bq-\bp)H(f_2|U) + (\lambda \q-\p)H(f_1|U) \leq 0$ for any fixed $p(x)$ with equality if $U=X$. Thus,
\begin{align*}
\max_{(R_1,R_2) \in \bar{\Rr}}(R_1+\lambda R_2) &\le \max_{p(x)} (\p H(f_1)+\bp H(f_2)).
\end{align*}
Since $(C_1, 0) = (\max_{p(x)} (\p H(f_1)+\bp H(f_2)), 0) \in \Cr$. the inequality~\eqref{ineq1} holds.
\item For $\bp/\bq<\lambda \leq 1$, consider 
\begin{align*}
(\lambda\bq-\bp)H(f_2|U) + (\lambda \q-\p)H(f_1|U)
&=(\lambda-1)H(f_1|U) + (\lambda\bq-\bp) (H(f_2|f_1,U)-H(f_1|f_2,U))\nonumber\\
&\leq (\lambda\bq-\bp)H(f_2|f_1)
\end{align*}
for any fixed $p(x)$ with equality if $U=f_1$. Thus,
\begin{align*}
\max_{(R_1,R_2) \in \bar \Rr}(R_1+\lambda R_2) &\le \max_{p(x)}(\p H(f_1)+\bp H(f_2)+(\lambda \bq-\bp)H(f_2|f_1))\\
&= \max_{p(x)}(I(f_1;Y_1|S) +\l (I(f_2;Y_2|S) - I(f_1;f_2)))\\
&= \max_{p(x) \in \Pc}(I(f_1;Y_1|S) +\l (I(f_2;Y_2|S) - I(f_1;f_2)))
\end{align*}
Finally since $(I(f_1;Y_1|S), I(f_2;Y_2|S) - I(f_1;f_2)) \in \Cr$ for $p(x) \in \Pc$, the inequality~\eqref{ineq1} holds.
\end{itemize}

We now prove the inequality~\eqref{ineq1} for $\l > 1$. We consider the equivalent maximization problem: $\max_{(R_1,R_2) \in \bar \Rr}(\lambda^{-1}R_1+ R_2)$. Consider\begin{align}
\max_{(R_1,R_2) \in \bar \Rr}(\lambda^{-1}R_1+R_2)&\le \max_{p(v,x)} (\lambda^{-1} H(Y_1|V,S)+I(V;Y_2|S))\nonumber\\
& =\max_{p(x)}\Big(H(Y_2|S)+\max_{p(v|x)}(\lambda^{-1} H(Y_1|V,S) - H(Y_2|V,S)) \Big)\nonumber\\
&=\max_{p(x)}\Big(\q H(f_1)+ \bq H(f_2) + \max_{p(v|x)}((\lambda^{-1} \bp-\bq)H(f_2|V) + (\lambda^{-1}\p-\q)H(f_1|V)) \Big).\nonumber
\end{align}
We now consider different ranges of $\lambda > 1$.
\begin{itemize}
\item For $\lambda > \p/\q$, $(\lambda^{-1}\bp-\bq)H(f_2|V) + (\lambda^{-1} \p-\q)H(f_1|V) \leq 0$ for any fixed $p(x)$ with equality if $V=X$. Thus,
\begin{align*}
\max_{(R_1,R_2) \in \bar \Rr}(R_1+\lambda R_2) \le \max_{p(x)}\ (\lambda \q H(f_1)+\lambda \bq H(f_2)).
\end{align*}
Since $(0,C_2) = (0, \max_{p(x)} (\q H(f_1)+\bq H(f_2))) \in \Cr$. the inequality~\eqref{ineq1} holds.
\item For $1 < \lambda \leq \p/\q$, consider
\begin{align*}
(\lambda^{-1}\bp-\bq)H(f_2|V) + (\lambda^{-1}\p-\q)H(f_1|V)&=(\lambda^{-1}-1)H(f_2|V)+(\lambda^{-1}\p-\q)(H(f_1|f_2,V)-H(f_2|f_1,V))\nonumber\\
&\leq (\lambda^{-1}\p-\q)H(f_1|f_2)
\end{align*}
for any fixed $p(x)$ with equality if $V=f_2$. Thus,
\begin{align*}
\max_{(R_1,R_2) \in \bar{\Rr}}(R_1+ \lambda R_2) &\le \max_{p(x)}\big(\l p_2 H(f_1) + \l \bar{p}_2 H(f_2) + (p_1 - \l p_2) H(f_1|f_2)\big)\\
& = \max_{p(x)} (I(f_1;Y_1|S) - I(f_1;f_2) +\l I(f_2;Y_2|S))\\
& = \max_{p(x) \in \Pc} (I(f_1;Y_1|S) - I(f_1;f_2) +\l I(f_2;Y_2|S)).
\end{align*}
Finally since $(I(f_1;Y_1|S) - I(f_1;f_2),I(f_2;Y_2|S)) \in \Cr$ for $p(x) \in \Pc$, the inequality~\eqref{ineq1} holds.
\end{itemize}

The proof of the converse is completed using the following lemma.
\end{proof}

\begin{lemma}\label{subseteq}
\textnormal{~\cite{Eggleston1958} Let $\Rr \in \mathbb{R}^d$ be convex and $\Rr_1 \subseteq \Rr_2$ be two bounded convex subsets of $\Rr$, closed relative to $\Rr$. If every supporting hyperplane of $\Rr_2$ intersects $\Rr_1$, then $\Rr_1 = \Rr_2$.
}\end{lemma}
\smallskip

As an example of a BC-TDCS, consider the following. 
\begin{example}[Blackwell channel with state~\cite{Kim--Chia--EG2015}]\label{ex:blackwell}
\textnormal{The functions $f_1$ and $f_2$ for this example are depicted in Figure~\ref{fig:bc-blackwell}.
\begin{figure}[h]
\begin{center}
\psfrag{0}[r]{0}
\psfrag{1}[r]{2}
\psfrag{2}[r]{1}
\psfrag{X}[c]{$X$}
\psfrag{Y1}[c]{$f_1(X)$}
\psfrag{Y2}[c]{$f_2(X)$}
\psfrag{z}[l]{0}
\psfrag{o}[l]{1}
\includegraphics[scale=0.46]{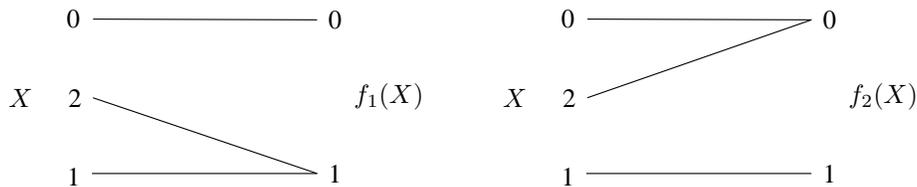}
\end{center}
\caption{The deterministic components of the Blackwell channel with state.} 
\label{fig:bc-blackwell}
\end{figure}}

\textnormal{The private message capacity region of the Blackwell channel with state known only to the receivers is the convex hull of
\begin{align*}
\Rr'_1 =\{(R_1,R_2)\colon\, &R_1 \leq H(\a_0) - \bp \bar{\a}_1 H(\a_0/\bar{\a}_1),\\
						&R_2 \leq H(\a_1) - \q\bar{\a}_0H(\a_1/\bar{\a}_0),\\
						&R_1+R_2 \leq H(\a_0) - \bp \bar{\a}_1 H(\a_0/\bar{\a}_1)+\bq\bar{\a}_0 H(\a_1/\bar{\a}_0)\\
						&\text{for some } \a_0,\a_1\geq 0, \a_0+\a_1 \leq 1\}.
\end{align*}
where $H(a)$, $a \in [0,1]$ is the binary entropy function. Note that $\Rr'_1$, defined in the proof of Theorem~\ref{cap:private}, is the Marton rate region with $(U_1,U_2) = (f_1,f_2)$ and $X \in \{0,1,2\}$ for $p_X(0)=\a_0, p_X(1)=\a_1, p_X(2)=1-\a_0-\a_1$ for  $ \a_0,\a_1\geq 0, \a_0+\a_1 \leq 1$. Also, since the rate pairs $(C_1,0)=(1,0) \in \Rr'_1$ and $(0,C_2)=(0,1) \in \Rr'_1$, $\Cr$ is the convex hull of $\Rr'_1$. The capacity region with state for $(p_1,p_2)=(0.5, 0.5),\, (0.7,0.3)$, and $(1,0)$ is plotted in Figure~\ref{fig2}. For $(p_1,p_2)=(0.5, 0.5)$, the two channels are statistically identical, hence the capacity region coincides with the time-division region.  For $(p_1,p_2)=(1,0)$, the channel reduces to the Blackwell channel with no state~\cite{Blackwell--Breiman--Thomasian1958}. For $(p_1,p_2)$ in between these two extreme cases, the capacity region is established by our theorem.}
\begin{figure}[h]
  \centering
\psfrag{a1}[c]{0}
\psfrag{a2}[c]{}
\psfrag{b}[c]{1}
\psfrag{c}[c]{1}
\psfrag{e}[c]{}
\psfrag{f}[c]{}
\psfrag{n}[b][b][0.9]{$\qquad(0.5,0.5)$}
\psfrag{k}[c][l][0.9]{$\qquad(0.7,0.3)$}
\psfrag{m}[r][l][0.9]{$\qquad(1,0)$}
\psfrag{g}[b]{$R_2$}
\psfrag{h}[l]{$R_1$}
    \includegraphics[scale=0.5]{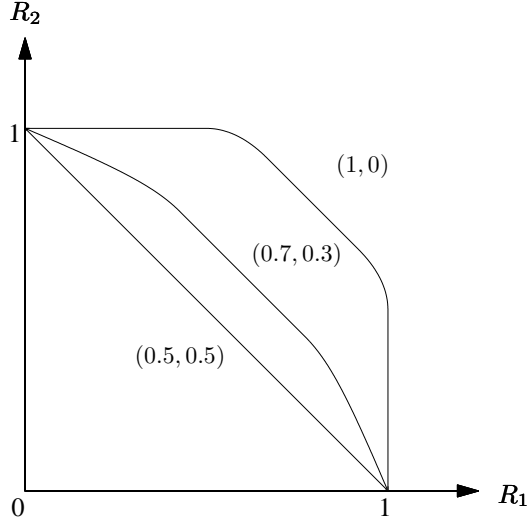}
      \caption{Capacity region of the Blackwell channel with the state. }
      \label{fig2}
 \end{figure}
 \end{example}

Next consider the following example which is motivated by deterministic approximations of wireless channels.

\begin{example}[Finite-field BC-TDCS] \label{ex:finite}
\textnormal{Consider the BC-TDCS  with the state known only at the receivers with $\mathbf{X}= \begin{bmatrix}X_1&X_2 \end{bmatrix}^T$:
\begin{align}\begin{split}\label{ex2}
Y_1 = \begin{cases}	h_{11}X_{1}+h_{12}X_2 &\text{ if } S_1=1,\\
				h_{21}X_{1}+h_{22}X_2 &\text{ if } S_1=2,
	\end{cases}\\
Y_2 = \begin{cases} h_{11}X_{1}+h_{12}X_2&\text{ if } S_2=1,\\
				h_{21}X_{1}+h_{22}X_2&\text{ if } S_2= 2,
	\end{cases}
\end{split}\end{align}
where the channel matrix is full-rank, $\Yc_1=\Yc_2=\Xc_1=\Xc_2=[0:K-1]$, and the arithmetic is over the finite field.}

\textnormal{To compute the private message capacity region, note that $C_1=\log K$ and  $C_2=\log K$. 
\textnormal{To evaluate $\Rr_1$, we compute $\Pc$. }
    \begin{align*}
\p H(f_1)+\bp I(f_1;f_2)+\lambda \bq H(f_2|f_1) &=\p H(f_1)+\bp H(f_2)+(\lambda \bq-\bp) H(f_2|f_1)\\
&\leq (\p+\lambda\bq)\log K
\end{align*}
for $\bp/\bq \leq \lambda \leq 1$ with equality if $\mathbf{X} \sim \U ([0:K-1]^2)$. Similarly,
\begin{align*}
\p H(f_1)+\bp I(f_1;f_2)+\lambda \bq H(f_2|f_1) &=\p H(f_1)+\bp H(f_2)+(\lambda \bq-\bp) H(f_2|f_1)\\
&\leq (\p+\lambda\bq)\log K 
\end{align*}
for $1 \leq \lambda \leq p_1/p_2$ with equality if $\mathbf{X} \sim \U ([0:K-1]^2)$. Thus, $\Pc = \left\{\U ( [0:K-1]^2) \right\}$.
Note that when $\mathbf{X}$ is uniform,  $H(f_{1})=H(f_{2})=H(f_{1}|f_{2})=H(f_{2}|f_{1})=\log K$. Hence,
    \begin{align*}
    \Rr_1 = \{(R_1,R_2)\colon &R_1 \leq \p\log K, R_2 \leq \bq \log K \},
    \end{align*}
and the capacity region is
\begin{align*}
\Cr =\co \{\Rr_1 \cup \{(\log K,0)\}\cup \{(0,\log K)\}\} = \co \{(0,0) \cup (\log K,0) \cup (0,\log K) \cup (\p \log K, \bq \log K)\}.
\end{align*}}

\textnormal{Figure~\ref{fig3} plots the capacity region for $(p_1,p_2)=(0.5,0.5), (0.7,0.4)$, and $(1,0)$. For $(p_1,p_2)=(0.5,0.5)$, the two channels are statistically identical and the capacity region coincides with the time-division region. For $(p_1,p_2)=(1,0)$, the capacity region is $\{(R_1,R_2)\colon R_1 \le \log K,\; R_2 \le \log K\}$ because the chahnnel matrix is full-rank. For $(p_1,p_2)$ in between these two extreme cases, the capacity region is established by our theorem.}
 \begin{figure}[h]
 \centering
 \psfrag{a1}[c]{0}
\psfrag{a2}[c]{}
\psfrag{b}[c]{1}
\psfrag{c}[r]{1}
\psfrag{e}[c]{0.7}
\psfrag{f}[r]{0.6}
\psfrag{n}[c][][0.9]{$\qquad(0.5,0.5)$}
\psfrag{k}[c][][0.9]{$\qquad(0.7,0.4)$}
\psfrag{m}[c][][0.9]{$\qquad(1,0)$}
\psfrag{g}[b]{$R_2/\log K$}
\psfrag{h}[l]{$R_1/\log K$}
 \includegraphics[scale=0.5]{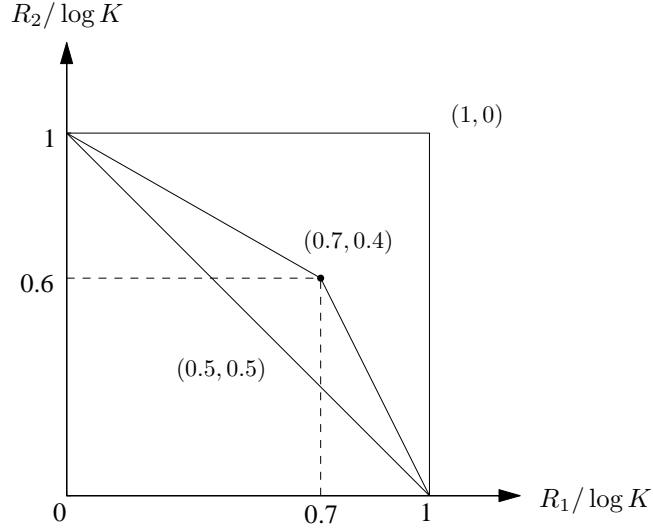}
 \caption{Capacity region of the Finite Field BC-TDCS.}\label{fig3}
 \end{figure}

\begin{remark}\it{Connection to wireless channels.} \textnormal{Consider the following fading broadcast channel
\begin{align}\label{fading}
Y_j=\mathbf{H}_j^{\dagger}\mathbf{X}+Z_j\text{ for } j=1,2,
\end{align}
where $\dagger$ denotes the conjugate-transpose, $\mathbf{X}=\begin{bmatrix}X_1&X_2 \end{bmatrix}^T\in \mathbb{C}^{2\times 1}, \E [\mathbf{X}^{\dagger}\mathbf{X}] \leq P$, $Z_j \sim \mathcal{C}\mathcal{N}(0,1)$ and the noise sequences $Z_{ji}$, $j=1,2$ and $i\in [1:n]$, are i.i.d. In addition, for $j=1,2$,
\begin{align*}
\mathbf{H}^{\dagger}_j =\begin{cases}
					[ h_{11} \quad h_{12}]& \text{if }S_j=1 \text{ w.p. } p_j,\\
					[h_{21}  \quad h_{22}] & \text{if }S_j=2 \text{ w.p. } \bar{p}_j,
					\end{cases}
\end{align*}%
where the channel matrix is in $\mathbb{C}^{2\times 2}$ and is full rank.}

\textnormal{We now show that the degrees of freedom (DoF) of this fading Gaussian broadcast channel, obtained by dividing the maximum sum-rate by $\log P$ and taking the limit, is $\p+\bq$.}

\textnormal{Since the variance of the noise $Z_j$ is bounded, the DoF of channel in~\eqref{fading} is equal to that of the BC-TDCS with $Y_j=\mathbf{H}_j^{\dagger}\mathbf{X}$ for $j=1,2$~\cite{Maddah-Ali--Tse2010}.}
\textnormal{We show that the DoF is achieved when $U_1=f_1$ and $U_2=f_2$ are independent and Gaussian with variances $\alpha P$ and $\beta P$ for some $\alpha,\beta>0$ such that
\begin{align*}
\begin{bmatrix}X_1\\X_2\end{bmatrix}=\begin{bmatrix}
    h_{11} & h_{12}\\
    h_{21} & h_{22}
    \end{bmatrix} ^{-1}\begin{bmatrix}U_1\\U_2\end{bmatrix}
\end{align*}
satisfy the power constraint. First note that for $(R_1,R_2) \in \Cr$,
\begin{align}
&\max \lim _{P \to \infty}\frac{R_1+R_2}{\log P}\nonumber\\
&=\max_{p(\mathbf{X})} \lim _{P \to \infty}\frac{\p H(f_1)+ \bq H(f_2)+(\bp-\bq)I(f_1;f_2)}{\log P}.\label{3terms}
\end{align}}

\textnormal{Now we show that each term in~\eqref{3terms} is maximized with the chosen input. First, $\lim _{P \to \infty}\p H(f_1)/\log P = \lim_{P \to \infty}\p \log(\alpha P)/\log P = p_1$. Now we show that $\p = \max \lim _{P \to \infty}\p H(f_1)/\log P$. Since  $\var(f_1)=\var(h_{11}X_1+h_{12}X_2) = |h_{11}|^2 \gamma P+|h_{12}|^2 \bar{\gamma} P+(h_{11}^*h_{12}+h_{12}^*h_{11})\rho \sqrt{\gamma \bar{\gamma}}P$ for some $0\leq \gamma,\rho \leq 1$ due to the power constraint, $H(f_1) \leq \log(|h_{11}|^2 \gamma+|h_{12}|^2 \bar{\gamma}+(h_{11}^*h_{12}+h_{12}^*h_{11})\rho\sqrt{\gamma \bar{\gamma}})+\log P$. Hence, $\lim _{P \to \infty}\p H(f_1)/\log P \leq \p$.
Similarly, $\lim _{P \to \infty}\bq H(f_2)/\log P$ is maximized and is equal to $\bq$, and $\lim _{P \to \infty}(\bp-\bq)I(f_1;f_2)/\log P$ is maximized and is equal to $0$. Thus, the following holds:
\begin{align*}
&\max_{p(\mathbf{X})} \lim _{P \to \infty}\frac{\p H(f_1)+ \bq H(f_2)+(\bp-\bq)I(f_1;f_2)}{\log P} \\
&\qquad= \p+\bq,
\end{align*}
and the DoF of the fading Gaussian BC in~\eqref{fading} is $\p+\bq$.}
\end{remark}
\end{example}

The capacity region result can be readily extended to the case with common message ($R_0 \neq 0$).
\begin{theorem}\label{capacity:Marton:common}\textnormal{The capacity region of a BC-TDCS $(\Xc \times \Sc, p(s)p(y_1,y_2|x,s), \Yc_1 \times \Yc_2)$ with the state known only at the receivers is the convex hull of the set of all rate pairs $(R_0,R_1,R_2)$ such that
\begin{align*}
R_0 &\leq \min\{I(U_0;Y_1|S), I(U_0;Y_2|S)\}\\
R_0+R_1 &\leq I(U_0;Y_1|S)+I(U_1;Y_1|U_0,S)\\
R_0+R_2 &\leq I(U_0;Y_2|S)+I(U_2;Y_2|U_0,S)\\
R_0+R_1+R_2 &\leq \min\{I(U_0;Y_1|S), I(U_0;Y_2|S)\}+I(U_1;Y_1|U_0,S)+I(U_2;Y_2|U_0,S)-I(U_1;U_2|U_0)\\
\end{align*}
for some $p(u_0,x)$ and either $(U_1,U_2)=(X,\emptyset), (U_1,U_2)=(f_1,f_2)$, or $(U_1,U_2)=(\emptyset,X)$.
}\end{theorem}
The proof is in Appendix~\ref{cap:common TDCS}.
\section{Ordered channel state components}\label{sec:Ordered}
Recall the definitions of the following classes of BC for which superposition coding was shown to be optimal.

\begin{definition}[Degraded BC~\cite{Cover1972}]\label{def1}\textnormal{For a DM-BC $(\Xc, p(\yt_1,\yt_2|x), \Ytcal_1 \times \Ytcal_2)$ receiver $\Yt_2$ is said to be a degraded version of $\Yt_1$ if there exists $Z$ such that $Z|\{X=x\} \sim p_{\Yt_1|X}(z|x),$ i.e., $Z$ has the same conditional pmf as $\Yt_1$ (given $X$), and $X \to Z \to \Yt_2$ form a Markov chain.  
}\end{definition} 

\begin{definition}[Less noisy BC~\cite{Korner--Marton1975}]\textnormal{For a DM-BC $(\Xcal, p(\yt_1,\yt_2|x), \Ytcal_1 \times \Ytcal_2)$ receiver $\Yt_1$ is said to be less noisy than $\Yt_2$ $I(U;\Yt_1) \geq I(U;\Yt_2)$ for all $p(u,x)$.
}\end{definition}
Van-Dijk~\cite{van-Dijk1997} showed that receiver $\Yt_1$ is less noisy than receiver $\Yt_2$ if $I(X;\Yt_1) - I(X;\Yt_2)$ is concave in $p(x)$, or equivalently, $I(X;\Yt_1)-I(X;\Yt_2)$ is equal to its upper concave envelope $\CC[I(X;\Yt_1) - I(X;\Yt_2)]$ (the smallest concave function that is greater than or equal to $I(X;\Yt_1)-I(X;\Yt_2)$). 

\begin{definition}[More capable BC~\cite{Korner--Marton1975}]\textnormal{For a DM-BC $(\Xcal, p(\yt_1,\yt_2|x), \Ytcal_1 \times \Ytcal_2)$ receiver $\Yt_1$ is said to be more capable than $\Yt_2$ if $I(X;\Yt_1) \geq I(X;\Yt_2)$ for all $p(x)$. 
}\end{definition}
The more capable condition can also be recast in terms of the concave envelope: Receiver $\Yt_1$ is more capable than $\Yt_2$ if $\CC[I(X;\Yt_2)-I(X;\Yt_1)]=0$ for every $p(x)$. 

\begin{definition}[Dominantly c-symmetric BC~\cite{Nair2010}]\label{def4} \textnormal{A DMC with input alphabet $\Xc = \{0,1,\ldots,m-1\}$ and output alphabet $\Yc$ of size $n$ is said to be {\em c-symmetric} if, for each $j = 0,\ldots,m-1,$ there is a permutation $\pi_j(\cdot)$ of $\Yc$ such that $p_{Y|X}(\pi_j(y)|(i+j)_m) = p_{Y|X}(y|i)$ for all $i$, where $(i+j)_m = (i+j) \text{ mod }m$. A DM-BC $(\Xcal, p(\yt_1,\yt_2|x), \Ytcal_1 \times \Ytcal_2)$ is said to be c-symmetric if both channel components $X$ to $\Yt_1$ and $X$ to $\Yt_2$ are c-symmetric. A c-symmetric DM-BC is said to be {\em dominantly c-symmetric} if
\begin{align}\label{condition:dominant}
I(X;\Yt_1)_p - I(X;\Yt_2)_p \leq I(X;\Yt_1)_u - I(X;\Yt_2)_u
\end{align}
for every $p(x)$, where $u(x)$ is the uniform pmf and $I(X;\Yt_1)_p$ denotes the mutual information between $X$ and $\Yt_1$ for $X \sim p(x)$.
}\end{definition}


In the following we show the surprising fact that if the DM-BC satisfies any of the above definitions, then the corresponding BC-TCS with the state known at the receivers also satisfies the same condition. Hence, the capacity regions for these corresponding BC-TCS are achieved using superposition coding.

\begin{theorem}\label{thm:dominance}
\textnormal{The DM-BC $(\Xcal, p(y_1,y_2,s|x), (\Ycal_1,\Scal) \times (\Ycal_2,\Scal))$ with state known only at the receivers is\\ 
(i) degraded if the DM-BC $p(\yt_1,\yt_2|x)$ is degraded,\\
(ii) less noisy if the DM-BC $p(\yt_1,\yt_2|x)$ is less noisy,\\
(iii) more capable if the DM-BC $p(\yt_1,\yt_2|x)$ is more capable, \\
(iv) dominantly c-symmetric if the DM-BC $p(\yt_1,\yt_2|x)$ is dominantly c-symmetric.
}\end{theorem}
\begin{proof}
We prove (i). The proof of the rest of this theorem is in Appendix~\ref{proof:thm:dominance}.
\begin{enumerate}
\item[(i)] For a degraded DM-BC $p(\yt_1,\yt_2|x)$, there exists $Z$ such that $Z|\{X=x\} \sim p_{\tilde{Y}_1|X}(z|x)$ and $X \rightarrow Z \rightarrow \tilde{Y}_2$ forms a Markov chain. We show that there exists $(Y_1',S_1',S_2')$ such that $(Y_1',S_1',S_2')|\{X=x\} \sim p_{Y_1,S_1,S_2|X}(y_1',s_1',s_2'|x)$ and $X \rightarrow (Y_1',S_1',S_2') \rightarrow (Y_2,S_1,S_2)$ forms a Markov chain.


Let $(S_1',S_2')$ be distributed according to
\begin{align*}
p_{S'_1|S_1,S_2}(s_1'|s_1,s_2) = \begin{cases}
q_1 &\text{ if } (s_1',s_2)=(1,1),\\
p_1 - q_1 &\text{ if } (s_1',s_2) = (1,2),\\
1-p_1 &\text{ if } (s_1',s_2) = (2,2),
\end{cases} 
\end{align*}
and $p_{S'_2|S'_1,S_1,S_2}(s'_2|s'_1, s_1, s_2) = p_{S_2|S_1}(s'_2|s'_1)$. Thus $(S_1',S_2')|\{X = x\} \sim p_{S_1,S_2|X}(s_1',s_2'|x)$ and  
\begin{align}\label{markov1}
p(s_1,s_2|y_1',s_1',s_2',x) = p(s_1,s_2|y_1',s_1',s_2').
\end{align}

Let $Y'_1$ be distributed according to
\begin{align*}
Y'_1 =\begin{cases}
Y_2 &\text{ if } (S'_1,S_2)=(1,1) \text{ or } (2,2),\\
Z &\text{ if } (S'_1,S_2)=(1,2)\end{cases} 
\end{align*}
where $X \rightarrow Z \rightarrow \Yt_2$.

Then $Y_1'|\{(S_1',S_2',X)=(s_1',s_2',x)\} \sim p_{Y_1|S_1,S_2,X}(y_1'|s_1',s_2',x)$ and 
\begin{align}\label{markov2}
p(y_2,|s_1,s_2,y_1',s_1',s_2',x) = p(y_2|s_1,s_2,y_1',s_1',s_2').
\end{align}

By~\eqref{markov1} and~\eqref{markov2}, It follows that $(Y_1',S_1',S_2')|\{X=x\} \sim p_{Y_1,S_1,S_2|X}(y_1',s_1',s_2'|x)$ and $X \rightarrow (Y_1',S_1',S_2') \rightarrow (Y_2,S_1,S_2)$ forms a Markov chain.
\end{enumerate}

\end{proof}

\begin{remark}\textnormal{The DM-BC $(\Xcal, p(y_1,y_2,s|x), (\Ycal_1,\Scal) \times (\Ycal_2,\Scal))$ is degraded, less noisy, or more capable \emph{if and only if} the DM-BC $p(\yt_1,\yt_2|x)$ is degraded, less noisy, or more capable, respectively (assuming $p_1 > p_2$).  If the DM-BC $(\Xcal, p(y_1,y_2,s|x), (\Ycal_1,\Scal) \times (\Ycal_2,\Scal))$ is degraded, there exists $(Y_1',S_1',S_2')$ such that $(Y_1',S_1',S_2')|\{X=x\} \sim p_{Y_1,S_1,S_2|X}(y_1',s_1',s_2'|x)$ and $X \rightarrow (Y_1',S_1',S_2') \rightarrow (Y_2,S_1,S_2)$ forms a Markov chain. Let $Z$ be distributed according to $p_{Z|X}(z|x) = p_{Y_1'|S_1',X}(z|1,x)$ and $p_{\Yt_2|Z,X}(\yt_2|z,x) = p_{Y_2|S_2,S_1',Y_1',X}(\yt_2|2,1,z,x)$. Then $Z|\{X=x\} \sim p_{\Yt_1|X}(z|x)$ and $\Yt_2|\{X=x\} \sim p_{\Yt_2|X}(\yt_2|x)$. Also $X \rightarrow Z \rightarrow \Yt_2$ because $p_{Y_2|S_2,S_1',Y_1',X}(\yt_2|2,1,z,x) = p_{Y_2|S_2,S_1',Y_1'}(\yt_2|2,1,z)$. Therefore the DM-BC $p(\yt_1,\yt_2|x)$ is degraded.
The proofs for less noisy and more capable DM-BC follow directly from the proof of part (ii) and (iii) of Theorem~\ref{thm:dominance}. We do not know however if the DM-BC $(\Xcal, p(y_1,y_2,s|x), (\Ycal_1,\Scal) \times (\Ycal_2,\Scal))$ is  dominantly c-symmetric if and only if the DM-BC $p(\yt_1,\yt_2|x)$ is dominantly c-symmetric.
}\end{remark}
It follows from Theorem~\ref{thm:dominance} that the capacity region of the BC-TCS satisfying the conditions in Theorem~\ref{thm:dominance} is the set of rate pairs $(R_1,R_2)$ such that
\begin{align}\begin{split}\label{sc region}
R_1 &\leq I(X;Y_1|U,S),\\
R_2 &\leq I(U;Y_2|S),\\	
R_1+R_2 &\leq I(X;Y_1|S)
\end{split}\end{align}
for some $p(u,x)$. 

\begin{remark}\label{rmk: superposition}
\textnormal{
Using superposition coding, receiver $\Yt_1$ can recover receiver $\Yt_2$'s message. Hence when there is common message ($R_0 \neq 0$), the capacity region is obtained by replacing $R_1$ with $R_0+R_1$. 
}\end{remark}

 As an example of BC-TCS with more capable or dominantly c-symmetric components, consider the following.
\begin{example}[A BC-TCS with a BSC and a BEC channel components]\label{ex:bscbec} \textnormal{A BC-TCS with a BSC and a BEC channel components has input $\Xcal = \{0,1\}$ and channel components BSC($p$) and BEC($e$). Without loss of generality, we assume $0 \leq p \leq 1/2$ and $0 \leq e \leq 1$. In~\cite{Nair2010}, it is shown that for the DM-BC $p(\yt_1,\yt_2|x)$, 
\begin{enumerate}
\item $\Yt_1$ is a degraded version of $\Yt_2$ if and only if $0 \leq e \leq 2p$.
\item $\Yt_2$ is less noisy than $\Yt_1$ if and only if $0 \leq e \leq 4p(1-p)$.
\item $\Yt_2$ is more capable than $\Yt_1$ if and only if $0\leq e \leq H(p)$.
\item $\Yt_1$ is dominantly c-symmetric if $H(p) \leq e \leq 1$.
\end{enumerate}
Hence, by Theorem~\ref{thm:dominance}, the corresponding BC-TCS with BSC($p$) and BEC($e$) channel components is degraded, less noisy, more capable, or dominantly c-symmetric for the above channel parameter ranges.
}\end{example}

\subsection{A product of reversely more capable channel components}
Another class of broadcast channel for which superposition coding is shown to be optimal for each component is the product of reversely more capable broadcast channels~\cite{Geng--Gohari--Nair--Yu2014}.

\begin{definition}[Product of reversely more capable BCs]\textnormal{A DM-BC $(\Xcal, p(\yt_1,\yt_2|x),\Ytcal_1 \times \Ytcal_2)$ is said to be a product of reversely more capable DM-BC if $\Xcal = (\Xcal_1, \Xcal_2)$, $\Ytcal_1 = (\Ytcal_{11}, \Ytcal_{12})$, $\Ytcal_2 = (\Ytcal_{21}, \Ytcal_{22})$, and $p(\yt_{11},\yt_{12},\yt_{21},\yt_{22}|x_1,x_2) = p(\yt_{11},\yt_{12}|x_1)p(\yt_{21},\yt_{22}|x_2)$, and $I(X_1; \Yt_{11}) \geq I(X_1; \Yt_{21})$ for all $p(x_1)$ and $I(X_2;\Yt_{12}) \leq I(X_2;\Yt_{22})$ for all $p(x_2)$.
}\end{definition}
We extend this definition to the broadcast channel with two channel state components as follows.

\begin{definition}[A product BC-TCS]\label{def:product}\textnormal{A 2-receiver {\em product broadcast channel with two channel state components} is a DM-BC with random sate $(\Xc \times \Sc, p(s)p(y_1,y_2|x,s), \Yc_1 \times \Yc_2)$, where $X=[X_1, X_2]$ and $S=(S_1,S_2)$ for 
\begin{align*}
Y_j&=[Y_{j1}, Y_{j2}],\\
S_j &= [S_{j1},S_{j2}],\\
Y_{j1} &= \begin{cases}	\Yt_{11} &\text{ if } S_{j1}=1,\\
				\Yt_{21} &\text{ if } S_{j1}=2
	\end{cases}\\
Y_{j2} &= \begin{cases}	\Yt_{12} &\text{ if } S_{j2}=1,\\
				\Yt_{22} &\text{ if } S_{j2}=2 
	\end{cases}
\end{align*}
for $j = 1,2$ and $p(\yt_{11},\yt_{12}|x_1)p(\yt_{21},\yt_{22}|x_2)$.
}\end{definition}

Let $p_{S_{j1}}(1) = p_{j1}$ and $p_{S_{j2}}(1) = p_{j2}$ for $j = 1,2$. Without loss of generality, we assume $p_{11} \ge p_{12}$ and $p_{21} \ge p_{22}$. In the following we establish the capacity region of BC-TCS with reversely more capable components.
\begin{theorem}\label{parallel}\textnormal{
A 2-receiver product BC-TCS $(\Xc \times \Sc, p(s)p(y_1,y_2|x,s), \Ycal_1 \times \Ycal_2)$ is more capable if the product DM-BC $(\Xcal,p(\yt_1,\yt_2|x), \Ytcal_1 \times \Ytcal_2)$ for $\Xcal = (\Xcal_1,\Xcal_2)$, $\Ytcal_1 = (\Ytcal_{11},\Ytcal_{12})$, $\Ytcal_2 = (\Ytcal_{21},\Ytcal_{22})$ is reversely more capable.
}\end{theorem}
\begin{proof}
We show that the product DM-BC $(\Xcal,p(y_1,y_2|x), \Ycal_1 \times \Ycal_2)$ for $\Xcal = (\Xcal_1,\Xcal_2)$, $\Ycal_1 = ((\Ycal_{11},\Scal_{11}),(\Ycal_{12},\Scal_{12}))$, $\Ycal_2 = ((\Ycal_{21},\Scal_{21}),(\Ycal_{22},\Scal_{22}))$ is reversely more capable if the product DM-BC $(\Xcal,p(\yt_1,\yt_2|x), \Ytcal_1 \times \Ytcal_2)$ for $\Xcal = (\Xcal_1,\Xcal_2)$, $\Ytcal_1 = (\Ytcal_{11},\Ytcal_{12})$, $\Ytcal_2 = (\Ytcal_{21},\Ytcal_{22})$ is reversely more capable.
Consider
\begin{align*}
I(X_1; Y_{11},S_{11}) &= p_{11} I(X_1;\Yt_{11}) + \bar{p}_{11} I(X_1;\Yt_{21})\\
					&\ge p_{12} I(X_1;\Yt_{11}) + \bar{p}_{12} I(X_1;\Yt_{21})\\
					&=I(X_1;Y_{21},S_{21}).
\end{align*}
Similarly we can show that $I(X_2;Y_{12},S_{12}) \leq I(X_2;Y_{22},S_{22})$.
\end{proof}

An immediate consequence of Theorem~\ref{parallel} is that superposition coding for each component is optimal. The capacity region is the region shown in Theorem 3 of~\cite{Geng--Gohari--Nair--Yu2014} by replacing $Y_{ji}$ with $(Y_{ji},S)$ for $j, i \in [1:2]$. 


\subsection{Gaussian vector channel components}\label{sec:Gaussian}
Consider the BC-TCS with degraded vector Gaussian channel components
\begin{align}\begin{split}\label{model:gaussian}
\mathbf{\Yt_1} &= G \mathbf{X} + \mathbf{Z_1},\\ 
\mathbf{\Yt_2} &= G \mathbf{X} + \mathbf{Z_2},
\end{split}\end{align}
where $\mathbf{X}, \mathbf{Z}_1, \mathbf{Z}_2 \in \mathbb{R}^t$ and $\mathbf{X}$ and $\mathbf{Z}_j$ are independent for $j=1,2$. The channel gain matrix is $G \in \mathbb{R}^{t \times t}$, and $\mathbf{Z_1} \sim \mathcal{N}(0,N_1)$, and $\mathbf{Z_2} \sim \mathcal{N}(0,N_2)$ for some $N_2 - N_1 \succeq 0$. Assume the average transmission power constraint $\sum_{i=1}^n \textbf{x}^T(m_1,m_2,i) \textbf{x}(m_1,m_2,i) \le nP$ for $(m_1,m_2) \in [1:2^{nR_1}] \times [1:2^{nR_2}]$.

By Theorem~\ref{thm:dominance}, the BC-TCS with degraded vector Gaussian channel components is degraded and its capacity region is achieved via superposition coding. In the following, we show that it suffices to consider only Gaussian $(U,X)$.

\begin{proposition}\label{prop:Gaussian}
\textnormal{
The capacity region of a BC-TCS with degraded vector Gaussian components is the set of rate pairs $(R_1, R_2)$ such that
\allowdisplaybreaks
\begin{align}\begin{split}\label{eq:proposition}
R_1 &\leq p_1 \log\frac{ | G K_1 G^T + N_1 | }{|N_1|}+ \bar{p}_1 \log \frac{ | G K_1 G^T + N_2 |}{|N_2|},\\
R_2 &\leq p_2 \log\frac{ | G K G^T + N_1 |}{| G K_1 G^T+ N_1 |} + \bar{p}_2 \log \frac{| G K G^T + N_2 |}{| G K_1 G^T+ N_2 |}
\end{split}\end{align}
for some $K \succeq 0$ for $\text{tr} (K) \le P$ and $K \succeq K_1 \succeq 0$.
}\end{proposition}

\begin{proof}
By Theorem~\ref{thm:dominance}, the capacity region is set of rate pairs $(R_1, R_2)$ such that
\begin{align}\begin{split}\label{C:degraded}
R_1 &\leq I(\mathbf{X};\mathbf{Y_1}|\mathbf{U},S),\\
R_2 &\leq I(\mathbf{U};\mathbf{Y_2}|S)
\end{split}\end{align}
for some pmf $p(\mathbf{u},\mathbf{x})$.

Let $\Cr_\mathrm{G}$ denote the set of rate pairs $(R_1,R_2)$ that satisfy the inequalities in~\eqref{C:degraded} for some $\mathbf{U} \sim \Ncal(0, K_1)$ and $\mathbf{V} \sim \Ncal(0,K-K_1)$, independent of each other, and $\mathbf{X}=\mathbf{U}+\mathbf{V}$ for some $K \succeq K_1 \succeq 0$ and $\text{tr} (K) \le P$. It can be easily shown that $\Cr_\mathrm{G}$ is the set of rate pairs that satisfy inequalities in~\eqref{eq:proposition}. To show that $\Cr_\mathrm{G}$ is the capacity region, we show the following. 
\begin{lemma}\label{lemma:Gaussian}
For all $\lambda \ge 0$,
\begin{align*}
\max_{(R_1,R_2)\in \Cr_\mathrm{G}} (R_1+\lambda R_2) \ge \max_{p(\mathbf{u},\mathbf{x})\colon \E[\mathbf{X}^T\mathbf{X}] \le P} (I(\mathbf{X};\mathbf{Y_1}|\mathbf{U},S)+\l I(\mathbf{U};\mathbf{Y_2}|S)).
\end{align*}
\end{lemma}
The proof of this lemma is in Appendix~\ref{proof:lemma:Gaussian}. The proof of Proposition~\ref{prop:Gaussian} is completed using Lemma~\ref{subseteq}.
\end{proof}

\begin{remark}\textnormal{Recall that Gaussian superposition coding and dirty paper coding both achieve the capacity region of Gaussian BC-TCS when $(p_1,p_2)=(1,0)$, i.e., when the channel gain is fixed. For general $(p_1,p_2)$, Gaussian superposition coding achieves the capacity region, but dirty paper coding does not. See Appendix~\ref{appendix:dpc} for the proof.
}\end{remark}
\section{More than two channel state components}\label{sec:More}
In this section we consider the BC with more than two channel state components. Consider a DM-BC with random state, where the state $S=(S_1, S_2)\in [1:k]^2$, $p_{S_1}(i)=p_i$ and $p_{S_2}(i)=q_i$, channel components $p(\yt_i|x)$ for $i \in [1:k]$, and outputs $Y_1 = \Yt_i$ if $S_1 = i$ and $Y_2 = \Yt_i$ if $S_2 = i$ for $i \in [1:k]$.


In the following we establish several results when $k>2$.
\subsection{Binary erasure broadcast channel with $k$ channel components}
Consider a BC with $k$ state components where the  channel $p(\yt_i|x)$ is a BEC($\e_i$), $0 \le \e_i \le 1$, for $i \in [1:k]$. We show that this channel is always less noisy.

\begin{theorem}\textnormal{The binary erasure broadcast channel with $k$ channel state components with the state known only at the receivers  is always less noisy.
}\end{theorem}
\begin{proof}
Without loss of generality, assume that the capacity of channel $p(y_1,s|x)$, $C_1$, is larger than the capacity of the channel $p(y_2,s|x)$, $C_2$. Then for any $p(u,x)$,
\begin{align*}
I(U;Y_1,S)
&=H(U)-\sum_{i=1}^{k}p_i \big(\e_i H(U)+(1-\e_i) H(U|X)\big)\\
&=C_1 I(U;X)\\
&\ge C_2 I(U;X)\\
&=I(U;Y_2,S).
\end{align*}
\end{proof}
An immediate consequence of this theorem is that the capacity region is achieved via
superposition coding. Since $I(U;Y_1|S) = C_1 I(U;X)$ and $I(X;Y_2|S) = H(X) - \sum_{i=1}^k p_i \e_i H(X) = C_2 H(X)$, superposition coding inner bound in~\eqref{sc region} is equivalent to the set of rate pairs that satisfy 
\begin{align*}
R_1 &\leq C_1 I(U;X),\\
R_2 &\leq C_2 H(X|U),\\
R_1+R_2 &\leq C_2 H(X)
\end{align*}
for some $p(u,x)$. 
It can be easily seen that any achievable rate pair $(R_1,R_2)$ satisfies $R_1/C_1+R_2/C_2 \leq H(X) \leq 1$, and the rate pairs $(C_1,0)$ and $(0,C_2)$ are achievable. Thus capacity region is the set of rate pairs $(R_1,R_2)$ such that
\begin{align*}
\frac{R_1}{C_1}+\frac{R_2}{C_2} \leq 1.
\end{align*}


\subsection{Binary symmetric broadcast channel with three channel components}
Consider a BC with three channel state components where the channel $p(\yt_i|x)$ is a BSC($\a_i$), $0 \le \a_i \le 1$, for $i \in [1:3]$. We can show that superposition coding is optimal for this channel.

%
%

\begin{theorem}\label{thm:bsbc}
\textnormal{The BC with three binary symmetric channel state components is more capable or dominantly c-symmetric.
} \end{theorem}

\begin{proof}
Let $D(x) = I(X;Y_1|S)-I(X;Y_2|S)$ for $X \sim \Bern(x)$, i.e.,
\begin{align*}
D(x) = \sum_{i=1}^{3}p_i H(x*\alpha_i)-\sum_{i=1}^{3} p_i H(\alpha_i)-\Big(\sum_{i=1}^{3}q_i H(x*\alpha_i)-\sum_{i=1}^{3} q_i H(\alpha_i)\Big),
\end{align*}
where $a*b = a(1-b) + b(1-a)$ for $a,b\in [0,1]$.
Without loss of generality, we assume $D(0.5)=C_1-C_2 \ge 0$. 

The DM-BC $(\Xcal, p(y_1,y_2,s|x), (\Ycal_1,\Scal) \times (\Ycal_2,\Scal))$ is dominantly c-symmetric if $X \to (Y_1,S)$ and $X \to (Y_2,S)$ are c-symmetric  and $I(X;Y_1|S)_p - I(X;Y_2|S)_p \leq I(X;Y_1|S)_u - I(X;Y_2|S)_u$.

Note that the proof of part (iv) of Theorem~\ref{thm:dominance} which shows that $X \to (Y_j,S)$ are c-symmetric if $X \to \Yt_j$ are symmetric for $j=1,2$ does not rely on the cardinality of $\Scal$. Thus the proof can be extended to show that  $X \to (Y_j,S)$ for $j=1,2$ are c-symmetric for BC with three channel state components. 
In order to show  the DM-BC $(\Xcal, p(y_1,y_2,s|x), (\Ycal_1,\Scal) \times (\Ycal_2,\Scal))$ is more capable or dominantly c-symmetric, we now show that $D(0.5) \ge D(x)$ for every $x \in [0,1]$ or $D(x) \ge 0$ for every $x \in [0,1]$. After some computation, we obtain
\begin{align*}
D''(x)= \frac{(p_1-q_1)((1-2\a_3)^2\a_1\bar{\a_1}-(1-2\a_1)^2\a_3\bar{\a_3})}{(x*\a_1)(1-x*\a_1)(x*\a_3)(1-x*\a_3)}+ \frac{(p_2-q_2)((1-2\a_3)^2\a_2\bar{\a_2}-(1-2\a_2)^2\a_3\bar{\a_3})}{(x*\a_2)(1-x*\a_2)(x*\a_3)(1-x*\a_3)}.
\end{align*}
Note that $D''(x) = 0$ if 
\[
\frac{(p_1-q_1)((1-2\a_3)^2\a_1\bar{\a_1}-(1-2\a_1)^2\a_3\bar{\a_3})}{(x*\a_1)(1-x*\a_1)}+ \frac{(p_2-q_2)((1-2\a_3)^2\a_2\bar{\a_2}-(1-2\a_2)^2\a_3\bar{\a_3})}{(x*\a_2)(1-x*\a_2)} = 0.
\]
Since $D''(x)=0$ has at most two solutions in $(0,1)$, $D'(x) = 0$ has at most three solutions in $(0,1)$. Since $D'(0.5) = 0$ and $D(x) = D(1-x)$, i.e., symmetric with respect to $x=0.5$, $D'(x)=0$ has one solution or three solutions. If it has one solution, $D(x)$ is concave (see Figure~\ref{fig:bssc}-(a) for an example). If it has three solutions, $D(x) \ge 0$ or $D(0.5) \ge D(x)$ for $x \in [0,1]$ as illustrated in Figure~\ref{fig:bssc}-(b) and~\ref{fig:bssc}-(c).

\begin{figure}[h]
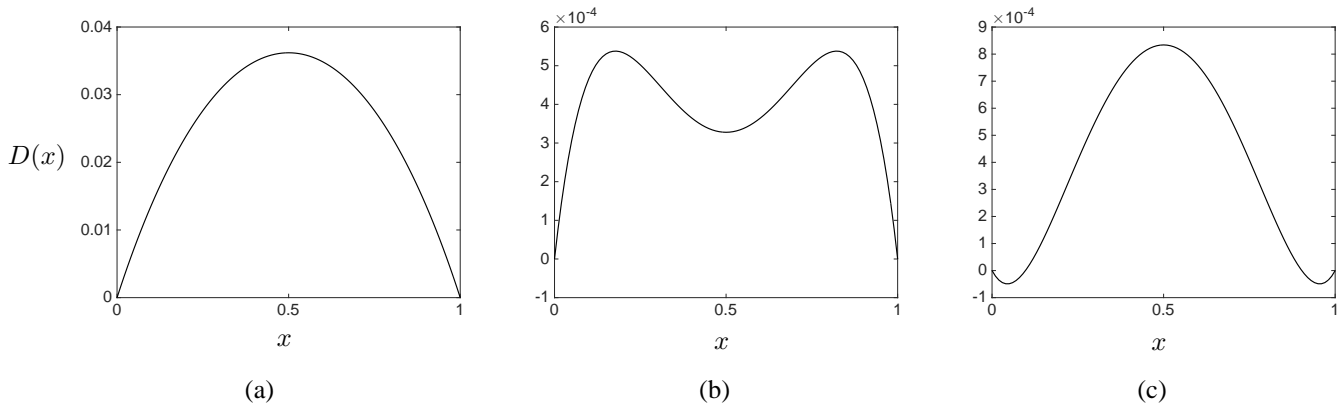

\begin{center}
\begin{tabular}{ccccccc}
\psfrag{a}[c]{(a)}
\psfrag{b}[c]{(b)}
\psfrag{c}[c]{(c)}
\psfrag{x}[cc]{$x$}
\psfrag{y}[c]{$D(x)\ \ \ \ \ \ \ \ $}
\ \ \includegraphics[scale=0.29]{bsbc1.eps}&&
\psfrag{x}[cc]{$x$}
\psfrag{y}[c]{}
\includegraphics[scale=0.29]{bsbc2.eps}&&
\psfrag{x}[cc]{$x$}
\psfrag{y}[c]{}
\includegraphics[scale=0.29]{bsbc3.eps}\\
\\
(a)&&(b)&&(c)
\end{tabular}
\end{center}
\caption{Examples of $D(x)$ vs $x$ for $\a = [0.2, 0.3, 0.4]$, $p=[1/3,1/3,1/3]$ and (a) $q=[0.2, 0.3, 0.5]$ (b) $q=[0.2, 0.7, 0.1]$ (c) $q=[0.45, 0, 0.55]$.}\label{fig:bssc}
\end{figure}


\end{proof}

An immediate consequence of this theorem is that the capacity region of the BC with three BSC state components when the state is known only at the receivers is the set of rate pairs $(R_1,R_2)$ that satisfy the inequalities in~\eqref{sc region}. We now show that this region is reduced to the set of rate pairs $(R_1,R_2)$ such that
\begin{align}\begin{split}\label{region:bsc}
R_1 &\le  \sum_{i=1}^{3} p_i \left(\sum_{j=1}^{2} \gamma_j H(\beta_j *\alpha_i)- H(\alpha_i) \right),\\ 
R_2 &\le 1 - \sum_{i=1}^{3} q_i\sum_{j=1}^{2} \gamma_j H(\beta_j*\alpha_i),\\
R_1+R_2 &\le 1-\sum_{i=1}^3 p_i H(\a_i)
\end{split}\end{align}
for some $0 \leq \gamma_j, \beta_j \leq 1$ for $j \in [1:2]$ such that $\gamma_1+\gamma_2 = 1$.

Suppose a rate pair $(R_1,R_2)$ satisfies the inequalities in~\eqref{sc region} for some $p(u,x)$ such that $|\Ucal|=3$. Then, this rate pair is also achievable with the following $(U',X')$ such that $U' \in \{-3,-2,-1,1,2,3\}$ and 
\begin{gather*}
  p_{U'}(u)=p_{U'}(-u)=\frac{1}{2} p_U(u), u \in \{1,2,3\},\\
  p_{X'|U'}(x|u)=p_{X'|U'}(1-x|-u)=p_{X|U}(x|u), (u,x) \in \{1,2,3\} \times \{0,1\}.
\end{gather*}

Further let $(Y'_1, Y'_2)$ be the output when the input is $X'$. It can be easily seen that $H(Y'_1|U',S)=H(Y_1|U,S)$ and $H(Y'_2|U',S)=H(Y_2|U,S)$. Also note that $H(Y_1'|S)=H(Y_2'|S) = 1$ because $X' \sim \Bern(1/2)$. Thus, $I(U;Y_1|S) \le I(U';Y_1'|S)$, $I(X;Y_2|U,S) \le I(X';Y_2'|U',S)$, and $I(X;Y_2|S) \le I(X';Y_2'|S)$.

Therefore, it suffices to evaluate the superposition rate region with the above symmetric input pmfs $p(u',x'),$ and the capacity region is the set of rate pairs $(R_1,R_2)$ that satisfy
\begin{align*}
R_1 &\le \sum_{i=1}^{3} p_i \left(\sum_{j=1}^{3} \gamma_j H(\beta_j *\alpha_i)- H(\alpha_i) \right),\\
R_2 &\le 1-\sum_{i=1}^{3}q_i \sum_{j=1}^3 \gamma_j H(\beta_j*\a_i),\\
R_1+R_2 &\le 1-\sum_{i=1}^3 p_i H(\a_i)
\end{align*}
for some $0 \leq \gamma_j, \beta_j \leq 1$ for $j \in [1:3]$ such that $\gamma_1+\gamma_2+\gamma_3 = 1$. Note that this rate region can be written as an intersection of two rate regions, $\Rr_1 \cap \Rr_2$, where $\Rr_1 = \{(R_1,R_2)\colon R_1+R_2 \le 1-\sum_{i=1}^3 p_i H(\a_i)\}$ and 
\begin{align*}
\Rr_2 = \Big\{(R_1,R_2):R_1 &\le \sum_{i=1}^{3} p_i \left(\sum_{j=1}^{3} \gamma_j H(\beta_j *\alpha_i)- H(\alpha_i) \right),\\
R_2 &\le 1-\sum_{i=1}^{3}q_i \sum_{j=1}^3 \gamma_j H(\beta_j*\a_i) \text{ for some } 0 \le \gamma_j, \beta_j \le 1 \text{ s.t. } \sum_{j=1}^3 \gamma_j = 1, j \in [1:3]\Big\}.
\end{align*}
Let $\bar \Rr_{2}$ denote the \emph{convex hull} of the set of rate pairs $(R_1,R_2)$ such that 
\begin{align*}
R_1 &\le \sum_{i=1}^{3} p_i \left(H(\beta *\alpha_i)- H(\alpha_i) \right),\\
R_2 &\le 1-\sum_{i=1}^{3}q_i H(\beta*\a_i)
\end{align*}
for some $0 \le \beta \le 1$. Note that since $\bar \Rr_2$ is a convex set in 2-dimension, all rate pairs in $\bar \Rr_2$ is a convex combination of two rate pairs included in $\bar \Rr_2$. Thus,
\begin{align*}
\bar \Rr_2 = \Big\{(R_1,R_2):R_1 &\le \sum_{i=1}^{3} p_i \left(\sum_{j=1}^{2} \gamma_j H(\beta_j *\alpha_i)- H(\alpha_i) \right),\\
R_2 &\le 1-\sum_{i=1}^{3}q_i \sum_{j=1}^2 \gamma_j H(\beta_j*\a_i) \text{ for some } 0 \le \gamma_j, \beta_j \le 1 \text{ s.t. } \sum_{j=1}^2 \gamma_j = 1, j \in [1:2]\Big\}.
\end{align*}
Note that this rate region is a subset of $\Rr_2$, so $\bar \Rr_2 \subseteq \Rr_2$. Also it can be easily seen that $\Rr_2 \subseteq \bar \Rr_2$, and so $\Rr_2 = \bar \Rr_2$. Therefore, the capacity region for BC-TCS with three BSC channel components is $\Rr_1 \cap \bar \Rr_2$, the region shown in~\eqref{region:bsc}.

In the following we show that superposition coding is not in general optimal for BC with more than three BSC state components.
\subsection{Binary symmetric broadcast channel with four channel components}
Consider a BC-TCS with BSC components with $\alpha_1 =0.28, \alpha_2 =0.04, \alpha_3=0.02, \alpha_4=0.18$, and $p=[0.38,0.62,0,0]$ and $q=[0,0,0.38,0.62]$. Thus $C_1=0.5247$, $C_2=0.5246$, and the maximum sum rate for superposition coding is $\max(C_1,C_2) = 0.5247$.

Now we consider the Marton inner bound. In~\cite{Geng--Jog--Nair--Wang2013}, Geng, Jog, Nair and Wang showed that for binary input broadcast channels, Marton's inner bound reduces to the set of rate pairs $(R_1,R_2)$ such that
\begin{align}\begin{split}\label{marton-inner}
R_1 &< I(W;Y_1) + \sum_{j=1}^k\beta_j I(X;Y_1|W=j),\\
R_2 &< I(W;Y_2) + \sum_{j=k+1}^5 \beta_j I(X;Y_2|W=j),\\
R_1 + R_2 &< \min\{I(W;Y_1), I(W;Y_2)\} +  \sum_{j=1}^k\beta_j I(X;Y_1|W=j) + \sum_{j=k+1}^5 \beta_j I(X;Y_2|W=j)
\end{split}\end{align}
for some $p_W(j)=\beta_j$, $j\in [1:5]$, and $p(x|w)$. This region is achieved using \emph{randomized time-division}~\cite{Hajek--Pursley1979}. 
This ingenious insight helps simplify the computation of Marton's inner bound for BC-TCS with BSC components. 
In this case, the maximum sum rate is $0.5250$ and is strictly greater than maximum sum rate for superposition coding. 
Thus, superposition coding is suboptimal. It is not known whether Marton coding is optimal, however, because there is a gap between the Marton maximum sum rate and the sum rate for the \emph{UV} outer bound, which in this case is at least 0.5256.

\section{Conclusion}
We established the capacity region of several classes of BC-TCS channel when the state is known only at the receivers. 
When the channel state components are deterministic, the capacity region is achieved via Marton coding. This is an interesting result because this channel model does not belong to any class of broadcast channels for which the capacity was previously known. When the channel state components are ordered, the BC-TCS is also ordered and the capacity region is achieved via superposition coding. We showed that when the BC-TCS has degraded vector Gaussian channel components, the capacity region is attained via Gaussian input and auxiliary random variables. We extended our results on ordered channel components to two example channels with more than two channel components, but showed that this extension does not hold in general.
\bibliographystyle{IEEEtran}
\bibliography{nit}
\appendices

\section{Proof of theorem~\ref{capacity:Marton:common}}\label{cap:common TDCS}
Let $\Cr_\mathrm{o}$ denote the region shown in Theorem~\ref{capacity:Marton:common}. Achievability follows immediately since $\Cr_\mathrm{o}$ is included in Marton's inner bound with common message. 

To establish the converse, we show that the capacity region coincides with the \emph{UVW} outer bound. The \emph{UVW} outer bound for the broadcast channel with state known at the receivers states that if a rate tuple $(R_0,R_1,R_2)$ is achievable, then it must satisfy the inequalities
\begin{align*}
R_0 &\le \min\{I(U_0;Y_1|S), U(U_0;Y_2|S)\},\\
R_0+R_1 &\le I(U_1;Y_1|U_0,S) + \min\{I(U_0;Y_1|S), U(U_0;Y_2|S)\},\\
R_0+R_2 &\le I(U_2;Y_2|U_0,S) + \min\{I(U_0;Y_1|S), U(U_0;Y_2|S)\},\\
R_0+R_1+R_2 &\le \min\{I(U_0;Y_1|S), U(U_0;Y_2|S)\} + I(U_1;Y_1|U_2,U_0,S) + I(U_2;Y_2|U_0,S),\\
R_0+R_1+R_2 &\le \min\{I(U_0;Y_1|S), U(U_0;Y_2|S)\} + I(U_1;Y_1|U_0,S) + I(U_2;Y_2|U_1,U_0,S)
\end{align*}
for some pmf $p(u_0,u_1,u_2,x)$. Let this outer bound be denoted by $\bar{\Rr}_\mathrm{o}$. We now show that every supporting hyperplane of $\bar \Rr_\mathrm{o}$ intersects $\Cr_\mathrm{o}$, i.e.
\begin{align}\label{ineq0}
\max_{(R_0,R_1,R_2) \in \bar \Rr_\mathrm{o}} (\l_0 R_0 + \l_1 R_1 + \l_2 R_2) \le \max_{(r_0,r_1,r_2) \in \Cr_\mathrm{o}} (\l_0 r_0 + \l_1 r_1 + \l_2 r_2).
\end{align}
We consider different ranges of $(\l_0,\l_1,\l_2)$ and show that the inequality~\eqref{ineq0} always holds.
\begin{enumerate}
\item[(1)] If $\l_2 \le \l_0 \le \l_1$ or $\l_0 \le \l_2 \le \l_1$, note that for any $(R_0,R_1,R_2) \in \bar \Rr_\mathrm{o}$,
\[
\l_0 R_0 + \l_1 R_1 + \l_2 R_2 \le \l_1 (R_0+R_1) + \l_2 R_2.
\] 
Thus
\begin{align*}
\max_{(R_0,R_1,R_2) \in \bar \Rr_\mathrm{o}} (\l_0 R_0 + \l_1 R_1 + \l_2 R_2) &\le \max_{(r_1,r_2) \in \bar \Rr} (\l_1 r_1 + \l_2 r_2)\\
&\le \max_{(r_0,r_1,r_2) \in \Cr_\mathrm{o}} (\l_0 r_0 + \l_1 r_1 + \l_2 r_2),
\end{align*}
where $\bar \Rr$ denotes the \emph{UV} outer bound in~\eqref{uv}. The last inequality follows because $\Cr_\mathrm{o}$ includes the private message capacity region $\Cr$.

\item[(2)] If $\l_1 \le \l_0 \le \l_2$ or $\l_0 \le \l_1 \le \l_2$, 
 note that for any $(R_0,R_1,R_2) \in \bar \Rr_\mathrm{o}$,
\[
\l_0 R_0 + \l_1 R_1 + \l_2 R_2 \le \l_1 R_1 + \l_2 (R_0+R_2).
\] 
Thus
\begin{align*}
\max_{(R_0,R_1,R_2) \in \bar \Rr_\mathrm{o}} (\l_0 R_0 + \l_1 R_1 + \l_2 R_2) &\le \max_{(r_1,r_2) \in \bar \Rr} (\l_1 r_1 + \l_2 r_2)\\
&\le \max_{(r_0,r_1,r_2) \in \Cr_\mathrm{o}} (\l_0 r_0 + \l_1 r_1 + \l_2 r_2),
\end{align*}
where $\bar \Rr$ denotes the \emph{UV} outer bound in~\eqref{uv}. The last inequality follows because $\Cr_\mathrm{o}$ includes the private message capacity region $\Cr$.
\item[(3)] If $\l_1 \le \l_2 \le \l_0$, note that for any $(R_0,R_1,R_2) \in \bar \Rr_\mathrm{o}$,
\begin{align*}
\lambda_0 R_0 + \lambda_2 R_2 + \lambda_1 R_1 &\leq \lambda_0 \min\{I(U_0;Y_1|S), I(U_0;Y_2|S)\} + \lambda_2 I(U_2;Y_2|U_0,S) + \lambda_1 I(X;Y_1|U_2,U_0,S)\\
&= \lambda_0 \min\{I(U_0;Y_1|S), I(U_0;Y_2|S)\} + \lambda_2 I(U_2;Y_2|U_0,S) + \lambda_1 H(Y_1|U_2,U_0,S) \\
&= \lambda_0 \min\{I(U_0;Y_1|S), I(U_0;Y_2|S)\} + \lambda_2 H(Y_2|U_0,S) \\
&\qquad + \lambda_1 H(Y_1|U_2,U_0,S)-\lambda_2 H(Y_2|U_2,U_0,S)\\
&= \lambda_0 \min\{I(U_0;Y_1|S), I(U_0;Y_2|S)\} + \lambda_2 p_2 H(f_1|U_0)+ \l_2 \bar{p}_2 H(f_2|U_0) \\
&\qquad+(\bar{p}_1\lambda_1-\bar{p}_2\lambda_2) H(f_2|U_2,U_0)+ (p_1\lambda_1-p_2\lambda_2)H(f_1|U_2,U_0)
\end{align*}
For a fixed $p(u_0,x)$, only the last two terms depend on $p(u_2|u_0,x)$. We now consider different ranges of $(\l_1,\lambda_2)$.
\begin{itemize}
\item If $\lambda_2 \geq p_1\lambda_1/p_2$, then for any fixed $p(u_0,x)$,
\begin{align*}
(\bar{p}_1\lambda_1-\bar{p}_2\lambda_2) H(f_2|U_2,U_0)+ (p_1\lambda_1-p_2\lambda_2)H(f_1|U_2,U_0) \leq 0
\end{align*}
with equality if $U_2=X$. Thus,
\begin{align*}
\lambda_0 R_0 + \lambda_2 R_2 + \lambda_1 R_1 &\leq \lambda_0 \min\{I(U_0;Y_1|S), I(U_0;Y_2|S)\} + \lambda_2 p_2 H(f_1|U_0)+ \l_2 \bar{p}_2 H(f_2|U_0)\\
&\le \max_{(r_0,r_1,r_2) \in \Cr_\mathrm{o}} (\lambda_0 r_0 + \lambda_2 r_2 + \lambda_1 r_1).
\end{align*}
\item If $\lambda_2 < p_1\lambda_1/p_2$, then for any fixed $p(u_0,x)$,
\begin{align*}
&(\bar{p}_1\lambda_1-\bar{p}_2\lambda_2) H(f_2|U_2,U_0)+(p_1\lambda_1-p_2\lambda_2)H(f_1|U_2,U_0) \\
&\qquad= (\lambda_1-\lambda_2)H(f_2|U_2,U_0) + (p_1\lambda_1-p_2\lambda_2)\{H(f_1|U_2,U_0)-H(f_2|U_2,U_0)\}\\
&\qquad= (\lambda_1-\lambda_2)H(f_2|U_2,U_0) + (p_1\lambda_1-p_2\lambda_2)\{H(f_1|f_2,U_2,U_0)-H(f_2|f_1,U_2,U_0)\}\\
&\qquad\leq (p_1\lambda_1-p_2\lambda_2)H(f_1|f_2,U_0)
\end{align*}
with equality if $U_2=f_2$. Thus, 
\begin{align*}
&\lambda_0 R_0 + \lambda_2 R_2 + \lambda_1 R_1 \leq \lambda_0 \min\{I(U_0;Y_1|S), I(U_0;Y_2|S)\} + \lambda_2 H(Y_2|U_0,S) + (p_1\lambda_1-p_2\lambda_2)H(f_1|f_2,U_0)\\
&\qquad= \lambda_0 \min\{I(U_0;Y_1|S), I(U_0;Y_2|S)\} + \lambda_2 p_2 H(f_1|U_0) +\lambda_2 \bar{p}_2 H(f_2|U_0)+(p_1\lambda_1-p_2\lambda_2)H(f_1|f_2,U_0)\\
&\qquad= \lambda_0 \min\{I(U_0;Y_1|S), I(U_0;Y_2|S)\} + \lambda_2 I(f_2;Y_2|U_0,S)+\lambda_1 p_1 H(f_1|f_2,U_0)\\
&\qquad= \lambda_0 \min\{I(U_0;Y_1|S), I(U_0;Y_2|S)\} + \lambda_2 I(f_2;Y_2|U_0,S)+\lambda_1 (I(f_1;Y_1|U_0,S) - I(f_1;f_2|U_0,S))\\
&\qquad\le \max_{(r_0,r_1,r_2) \in \Cr_\mathrm{o}} (\lambda_0 r_0 + \lambda_2 r_2 + \lambda_1 r_1).
\end{align*}
\end{itemize}
\item[(4)] If $\l_2 \le \l_1 \le \l_0$, note that for any $(R_0,R_1,R_2) \in \bar{\Rr}_\mathrm{o}$,
\begin{align*}
\lambda_0 R_0 + \lambda_1 R_1 + \lambda_2 R_2 &\leq \lambda_0 \min\{I(U_0;Y_1|S), I(U_0;Y_2|S)\} + \lambda_1 I(U_1;Y_1|U_0,S) + \lambda_2 I(X;Y_2|U_1,U_0,S)\\
&= \lambda_0 \min\{I(U_0;Y_1|S), I(U_0;Y_2|S)\} + \lambda_1 I(U_1;Y_1|U_0,S) + \lambda_2 H(Y_2|U_1,U_0,S) \\
&= \lambda_0 \min\{I(U_0;Y_1|S), I(U_0;Y_2|S)\} + \lambda_1 H(Y_1|U_0,S) \\
&\qquad\qquad+ \lambda_2 H(Y_2|U_1,U_0,S)-\lambda_1 H(Y_1|U_1,U_0,S) \\
&= \lambda_0 \min\{I(U_0;Y_1|S), I(U_0;Y_2|S)\} + \lambda_1 p_1 H(f_1|U_0) + \bar{p}_1 H(f_2|U_0)\\
&\qquad\qquad+ (p_2\lambda_2-p_1\lambda_1)H(f_1|U_1,U_0)+(\bar{p}_2\lambda_2-\bar{p}_1\lambda_1) H(f_2|U_1,U_0)
\end{align*}
For a fixed $p(u_0,x)$, only the last two terms depend on $p(u_1|u_0,x)$. We now consider different ranges of $(\lambda_1,\l_2)$.
\begin{itemize}
\item If $\lambda_1 \geq \bar{p}_2\lambda_2/\bar{p}_1$, then for any fixed $p(u_0,x)$,
\begin{align*}
(p_2\lambda_2-p_1\lambda_1)H(f_1|U_1,U_0)+(\bar{p}_2\lambda_2-\bar{p}_1\lambda_1) H(f_2|U_1,U_0)\leq 0
\end{align*}
with equality if $U_1=X$. Thus,
\begin{align*}
\lambda_0 R_0 + \lambda_1 R_1 + \lambda_2 R_2 &\leq \lambda_0 \min\{I(U_0;Y_1|S), I(U_0;Y_2|S)\} +  \lambda_1 p_1 H(f_1|U_0) + \bar{p}_1 H(f_2|U_0)\\
&\le \max_{(r_0,r_1,r_2) \in \Cr_\mathrm{o}} (\lambda_0 r_0 + \lambda_2 r_2 + \lambda_1 r_1).
\end{align*}
\item If $\lambda_1 < \bar{p}_2\lambda_2/\bar{p}_1$, then for any fixed $p(u_0,x)$,
\begin{align*}
&(p_2\lambda_2-p_1\lambda_1)H(f_1|U_1,U_0)+(\bar{p}_2\lambda_2-\bar{p}_1\lambda_1) H(f_2|U_1,U_0)\\
&\qquad= (\lambda_2 - \lambda_1) H(f_1|U_1,U_0) + ((\bar{p}_2\lambda_2-\bar{p}_1\lambda_1)\{H(f_2|U_1,U_0)-H(f_1|U_1,U_0)\}\\
&\qquad= (\lambda_2 - \lambda_1) H(f_1|U_1,U_0) + ((\bar{p}_2\lambda_2-\bar{p}_1\lambda_1)\{H(f_2|f_1,U_1,U_0)-H(f_1|f_2,U_1,U_0)\}\\
&\qquad\leq (\bar{p}_2\lambda_2-\bar{p}_1\lambda_1)H(f_2|f_1,U_0)
\end{align*}
with equality if $U_1=f_1$. Thus, 
\begin{align*}
\lambda_0 R_0 + \lambda_1 R_1 + \lambda_2 R_2 &\leq \lambda_0 \min\{I(U_0;Y_1|S), I(U_0;Y_2|S)\} + \lambda_1 H(Y_1|U_0,S) + (\bar{p}_2\lambda_2-\bar{p}_1\lambda_1)H(f_2|f_1,U_0)\\
&= \lambda_0 \min\{I(U_0;Y_1|S), I(U_0;Y_2|S)\} \\
&\qquad\qquad+ \lambda_1 p_1 H(f_1|U_0) +\lambda_1 \bar{p}_1 H(f_2|U_0)+(\bar{p}_2\lambda_2-\bar{p}_1\lambda_1)H(f_2|f_1,U_0)\\
&= \lambda_0 \min\{I(U_0;Y_1|S), I(U_0;Y_2|S)\} + \lambda_1 I(f_1;Y_1|U_0,S)+\lambda_2 \bar{p}_2 H(f_2|f_1,U_0)\\
&= \lambda_0 \min\{I(U_0;Y_1|S), I(U_0;Y_2|S)\} + \lambda_1 I(f_1;Y_1|U_0,S)+\lambda_2 I(X;Y_2|f_1,U_0,S)\\
&\le \max_{(r_0,r_1,r_2) \in \Cr_\mathrm{o}} (\lambda_0 r_0 + \lambda_2 r_2 + \lambda_1 r_1).
\end{align*}
\end{itemize}
\end{enumerate}
The proof of the converse is completed using Lemma~\ref{subseteq}.

\section{Proof of \lowercase{(ii)} - \MakeLowercase{(iv)} of Theorem~\ref{thm:dominance}}\label{proof:thm:dominance}
We show that if a DM-BC $(\Xcal, p(\yt_1,\yt_2), \mathcal{\Yt}_1 \times \mathcal{\Yt}_2)$ is less noisy, more capable, or dominantly c-symmetric, then the DM-BC $(\Xcal, p(y_1,y_2,s|x), (\Ycal_1,\Scal) \times (\Ycal_2,\Scal))$ is also less noisy, more capable, or dominantly c-symmetric, respectively.

\begin{enumerate}
\item[(ii)] For a less noisy DM-BC $p(\yt_1,\yt_2|x)$, $I(U;\Yt_1) \ge I(U;\Yt_2)$ for every $p(u,x)$.  Consider 
\begin{align*}
I(U;Y_1,S)-I(U;Y_2,S) &= I(U;Y_1|S)-I(U;Y_2|S) \\
&= p_1 I(U;\tilde{Y}_1) + \bar{p}_1 I(U;\tilde{Y}_2) -p_2 I(U;\tilde{Y}_1) - \bar{p}_2 I(U;\tilde{Y}_2) \\
&=(p_1-p_2)(I(U;\tilde{Y}_1)-I(U;\tilde{Y}_2))\\
&\geq 0.
\end{align*}
Thus the DM-BC $(\Xcal, p(y_1,s_1,y_2,s_2|x), (\Ycal_1,\Scal) \times (\Ycal_2,\Scal))$ is less noisy. Note that if $p_1 > p_2$, $I(U;Y_1,S)-I(U;Y_2,S) \ge 0$ if and only if $I(U;\tilde{Y}_1)-I(U;\tilde{Y}_2) \ge 0$.
\item[(iii)] For a more capable DM-BC $p(\yt_1,\yt_2|x)$, $I(X;\Yt_1) \ge I(X;\Yt_2)$ for every $p(u,x)$. Consider 
\begin{align*}
I(X;Y_1,S)-I(X;Y_2,S) &= I(X;Y_1|S)-I(X;Y_2|S)\\
&= p_1 I(X;\tilde{Y}_1) + \bar{p}_1 I(X;\tilde{Y}_2) -  p_2 I(X;\tilde{Y}_1) - \bar{p}_2 I(X;\tilde{Y}_2)\\
&= (p_1-p_2) (I(X;\tilde{Y}_1) - I(X;\tilde{Y}_2))\\
&\geq 0.
\end{align*}
Thus the DM-BC $(\Xcal, p(y_1,s_1,y_2,s_2|x), (\Ycal_1,\Scal) \times (\Ycal_2,\Scal))$ is more capable. Note that if $p_1 > p_2$, $I(X;Y_1,S)-I(X;Y_2,S) \ge 0$ if and only if $I(X;\tilde{Y}_1)-I(X;\tilde{Y}_2) \ge 0$.
\item[(iv)] For a dominantly c-symmetric DM-BC $p(\yt_1,\yt_2|x)$, let $\pi^1_j(y)$ and $\pi^2_j(y)$ be functions that satisfy
\begin{align*}
p_{\Yt_1|X}(\pi^1_j(\yt_1)|(i+j)_m) &= p_{\Yt_1|X}(\yt_1|i),\\
p_{\Yt_2|X}(\pi^2_j(\yt_2)|(i+j)_m) &= p_{\Yt_2|X}(\yt_2|i)
\end{align*}
for $i \in [0:m-1]$, where $(i+j)_m$ denotes $(i+j)$ mod $m$. To show that the DM-BC $(\Xcal, p(y_1,s_1,y_2,s_2|x), (\Ycal_1,\Scal) \times (\Ycal_2,\Scal))$ is dominantly c-symmetric, we first show that $X \rightarrow (Y_1,S)$ and $X \rightarrow (Y_2,S)$ are c-symmetric channels. Let 
\begin{align*}
\pi_j(y,s) = (\pi^s_j(y),s).
\end{align*}
Consider
\begin{align*}
p_{(Y_1,S)|X}((y,s)|i) &= p_{S}(s)p_{Y_1|X,S}(y|i,s)\\
				     &= p_{S}(s)p_{\Yt_{s_1}|X}(y|i)\\
				     &= p_{S}(s)p_{\Yt_{s_1}|X}(\pi^{s_1}_j(y)|(i+j)_m)\\
				     &= p_{S}(s)p_{Y_1|X,S}(\pi^{s_1}_j(y)|(i+j)_m, s)\\
				     &= p_{(Y_1,S)|X}((\pi^{s_1}_j(y),s)|(i+j)_m)\\
				     &= p_{(Y_1,S)|X}(\pi_j(y,s)|(i+j)_m).
\end{align*}
Thus $X \rightarrow (Y_1,S)$ is c-symmetric. Similarly we can show that $X \rightarrow (Y_2,S)$ is c-symmetric. To complete the proof we show that the inequality in~\eqref{condition:dominant} holds. Consider 
\begin{align*}
I(X;Y_1,S)_p - I(X;Y_2,S)_p &= I(X;Y_1|S)_p - I(X;Y_2|S)_p\\
					&=(p_1-p_2)(I(X;\Yt_1)_p - I(X;\Yt_2)_p)\\
					&\leq (p_1-p_2)(I(X;\Yt_1)_u - I(X;\Yt_2)_u)\\
					&=I(X;Y_1,S)_u - I(X;Y_2,S)_u.
\end{align*}
Thus the DM-BC $(\Xcal, p(y_1,s_1,y_2,s_2|x), (\Ycal_1,\Scal) \times (\Ycal_2,\Scal))$ is dominantly c-symmetric. 

%

\allowdisplaybreaks
\end{enumerate}
\section{Proof of Lemma~\ref{lemma:Gaussian}}\label{proof:lemma:Gaussian}
We first prove the lemma for $\l \ge 1$. For $\lambda \ge 1 $, consider 
\allowdisplaybreaks
\begin{align*}
&\max_{p(\mathbf{u},\mathbf{x})\colon \atop \E[\mathbf{X}^T\mathbf{X}] \le P} \big(I(\mathbf{X};\mathbf{Y_1}|\mathbf{U},S) + \lambda I(\mathbf{U};\mathbf{Y_2}|S)\big)= \max_{p(\mathbf{x}) \colon \atop \E[\mathbf{X}^T\mathbf{X}] \le P} \big(\lambda I(\mathbf{X};\mathbf{Y_2}|S) + \CC[I(\mathbf{X};\mathbf{Y_1}|S) - \lambda I(\mathbf{X};\mathbf{Y_2}|S)]\big)\\
&\qquad= \max_{K\succeq 0 \colon \atop \mathrm{tr}(K) \le P} \max_{p(\mathbf{x}):  \atop\E[\mathbf{X}\mathbf{X}^T] = K} \big(\lambda I(\mathbf{X};\mathbf{Y_2}|S) + \CC[I(\mathbf{X};\mathbf{Y_1}|S) - \lambda I(\mathbf{X};\mathbf{Y_2}|S)]\big)\\
&\qquad= \max_{K\succeq 0 \colon \atop \mathrm{tr}(K) \le P} \max_{p(\mathbf{x}): \E[\mathbf{X}\mathbf{X}^T] = K} \big(\lambda I(\mathbf{X};\mathbf{Y_2}|S) + \CC[I(\mathbf{X};\mathbf{Y_1}|S) - \lambda I(\mathbf{X};\mathbf{Y_2}|S)]\big)\\
&\qquad= \max_{K\succeq 0 \colon \atop \mathrm{tr}(K) \le P} \max_{p(\mathbf{x}): \E[\mathbf{X}\mathbf{X}^T] = K} \big(\lambda p_2 I(\mathbf{X};\mathbf{\Yt_1}) + \lambda \bar{p}_2 I(\mathbf{X};\mathbf{\Yt_2}) + \CC[(p_1 - \l p_2)I(\mathbf{X};\mathbf{\Yt_1}) + (\bar{p}_1 - \l\bar{p}_2)I(\mathbf{X};\mathbf{\Yt_2})]\big)\\
&\qquad\le \max_{K\succeq 0 \colon \atop \mathrm{tr}(K) \le P} \Big(\max_{p(\mathbf{x}): \E[\mathbf{X}\mathbf{X}^T] = K} \lambda p_2 I(\mathbf{X};\mathbf{\Yt_1}) + \max_{p(\mathbf{x}): \E[\mathbf{X}\mathbf{X}^T] = K}  \lambda \bar{p}_2 I(\mathbf{X};\mathbf{\Yt_2}) \\
&\qquad\qquad \qquad \qquad+ \max_{p(\mathbf{x}): \E[\mathbf{X}\mathbf{X}^T] = K} \CC[(p_1 - \l p_2) I(\mathbf{X};\mathbf{\Yt_1}) + (\bar{p}_1 - \l\bar{p}_2)I(\mathbf{X};\mathbf{\Yt_2})]\Big)\\
&\qquad\stackrel{(a)}{=} \max_{K\succeq 0 \colon \atop \mathrm{tr}(K) \le P} \Bigg(\lambda p_2 \log\frac{|G K G^T + N_1|}{|N_1|} + \lambda \bar{p}_2 \log\frac{|G K G^T + N_2|}{|N_2|}\\
&\qquad\qquad \qquad \qquad+  \max_{K_1 \colon \atop 0 \preceq K_1 \preceq K} \Big((p_1 - \l p_2) \log \frac{|G K_1 G^T + N_1|}{|N_1|} +(\bar{p}_1 - \l\bar{p}_2) \log \frac{|G K_1 G^T + N_2|}{|N_2|}\Big)\Bigg)\\
&\qquad= \max_{K\succeq 0 \colon \atop \mathrm{tr}(K) \le P}\max_{0 \preceq K_1 \preceq K} \Bigg(\lambda p_2 \log\frac{|G K G^T + N_1|}{|G K_1 G^T + N_1|} + \lambda \bar{p}_2 \log\frac{|G K G^T + N_2|}{|G K_1 G^T + N_2|} \\
&\qquad\qquad \qquad\qquad\qquad\qquad + p_1 \log \frac{|G K_1 G^T + N_1|}{|N_1|} + \bar{p}_1 \log \frac{|G K_1 G^T + N_2|}{|N_2|}\Bigg)\\
&\qquad= \max_{(R_1,R_2) \in \Cr_\mathrm{G}}(R_1+\l R_2).
\end{align*}
To show step $(a)$ for $\lambda \ge p_1/p_2$, note that $p_1-\l p_2 \le 0$ and $\bar{p}_1 - \l \bar{p}_2 \le 0$, and thus
\[
\max_{p(\mathbf{x})\colon \atop \E[\mathbf{X}\mathbf{X}^T] = K} \CC[(p_1 - \l p_2) I(\mathbf{X};\mathbf{\Yt_1}) + (\bar{p}_1 - \l\bar{p}_2)I(\mathbf{X};\mathbf{\Yt_2})] = 0.
\]
To show step $(a)$ for $1 \le \lambda < p_1/p_2$, note that $ (p_1 - \l p_2) I(\mathbf{X};\mathbf{\Yt_1}) + (\bar{p}_1 - \l\bar{p}_2)I(\mathbf{X};\mathbf{\Yt_2}) = (p_1 - \l p_2) (I(\mathbf{X};\mathbf{\Yt_1}) - \mu I(\mathbf{X};\mathbf{\Yt_2}))$ where $\mu = 1+(\l-1)/(p_1 - \l p_2) > 1$, and for $\mu > 1$,
\[
\max_{p(\mathbf{x})\colon \atop \E[\mathbf{X}\mathbf{X}^T] = K} \CC[I(\mathbf{X};\Yt_1) - \mu I(\mathbf{X};\Yt_2)] = \max_{K_1 \colon K \succeq K_1}  \left(\log \frac{|G K_1 G^T + N_1|}{|N_1|} - \mu \log \frac{|G K_1 G^T + N_2|}{|N_2|}\right).
\]

We now prove the lemma for $\lambda < 1$. Since $I(\mathbf{U};\mathbf{Y_2}|S) \le I(\mathbf{U};\mathbf{Y_1}|S)$ for any $p(\mathbf{u},\mathbf{x})$, it follows that
\begin{align*}
 \max_{p(\mathbf{u},\mathbf{x})\colon \atop \E[\mathbf{X}^T\mathbf{X}] \le P} \big(I(\mathbf{X};\mathbf{Y_1}|\mathbf{U},S) + \l I(\mathbf{U};\mathbf{Y_2}|S)\big) & \le \max_{p(\mathbf{x})\colon \atop \E[\mathbf{X}^T\mathbf{X}] \le P} I(\mathbf{X};\mathbf{Y_1}|S)\\
&= \max_{K\succeq 0 \colon \mathrm{tr}(K) \le P} \Big(p_1\log \frac{|G K G^T + N_1|}{|N_1|}+\bar{p}_1 \log \frac{|G K G^T + N_2|}{|N_2|}\Big)\\
& \le \max_{(R_1,R_2) \in \Cr_\mathrm{G}} (R_1 + \lambda R_2).
\end{align*}

\section{suboptimality of dirty paper coding}\label{appendix:dpc}
Consider a BC-TCS with scalar Gaussian channel components, i.e., $t=1$ in~\eqref{model:gaussian}. Let $U_1 \sim \mathcal{N}(0,1), U_2 \sim \mathcal{N}(0,1)$, and $\E[U_1 U_2]=\rho$, $X=aU_1+bU_2$, where $(a,b,\rho)$ satisfies the power constraint, $\E[X^2]=a^2+b^2+2ab\rho = T \le P$. Using dirty paper coding, $(R_1,R_2)$ is achievable if 
\begin{align*}
R_{1} &< p_1\log\left(1+\frac{(a+b\rho)^2}{b^2(1-\rho^2)+N_1}\right)+\bar{p}_1\log\left(1+\frac{(a+b\rho)^2}{b^2(1-\rho^2)+N_2}\right),\nonumber \\
R_{2} &<p_2\log\left(1+\frac{(b+a\rho)^2}{a^2(1-\rho^2)+N_1}\right)+\bar{p}_2\log\left(1+\frac{(b+a\rho)^2}{a^2(1-\rho^2)+N_2}\right),\nonumber \\
R_{1}+R_{2} &< p_1\log\left(1+\frac{(a+b\rho)^2}{b^2(1-\rho^2)+N_1}\right)+\bar{p}_1\log\left(1+\frac{(a+b\rho)^2}{b^2(1-\rho^2)+N_2}\right)\\
&\qquad +p_2\log\left(1+\frac{(b+a\rho)^2}{a^2(1-\rho^2)+N_1}\right)+\bar{p}_2\log\left(1+\frac{(b+a\rho)^2}{a^2(1-\rho^2)+N_2}\right)\nonumber -\log\left(\frac{1}{1-\rho^2}\right).\\
\end{align*}
Let this region be denoted by $\Rr_\mathrm{D}$. To show that dirty paper coding is suboptimal, we show that $ \max_{(r_1,r_2) \in \Rr_\mathrm{D}} (r_1+\lambda r_2) < \max_{(R_1,R_2) \in \Cr_\mathrm{G}} (R_1+\lambda R_2)$ for some $\l > 1$. Note that for $\lambda > 1$,
\begin{align}
\max_{(R_1,R_2) \in \Rr_\mathrm{D}} (R_1+\lambda R_2) &= \max_{a,b,\rho\colon a^2+b^2+2ab\rho=T \le P} \Big(p_1\log\left(\frac{T+N_1}{b^2(1-\rho^2)+N_1}\right)+\bar{p}_1\log\left(\frac{T+N_2}{b^2(1-\rho^2)+N_2}\right)-\log\left(\frac{1}{1-\rho^2}\right)\nonumber\\
&\qquad\qquad\qquad\qquad\qquad\qquad+\lambda p_2\log\left(\frac{T+N_1}{a^2(1-\rho^2)+N_1}\right)+\lambda \bar{p}_2\log\left(\frac{T+N_2}{a^2(1-\rho^2)+N_2}\right)\Big),\nonumber\\
\max_{(R_1,R_2)\in \Cr_\mathrm{G}}(R_1+\lambda R_2) &= \max_{\alpha \in [0:1], T \le P} \Big(p_1 \log\left(\frac{\alpha T+N_1}{N_1}\right) + \bar{p}_1 \log\left(\frac{\alpha T+N_2}{N_2}\right)\nonumber\\
&\qquad\qquad\qquad\qquad+ \lambda p_2 \log\left(\frac{T+N_1}{\alpha T+N_1}\right) + \lambda\bar{p}_2 \log\left(\frac{T+N_2}{\alpha T+N_2}\right)\Big).\label{cg}
\end{align}
Consider
\begin{align*}
&\max_{(R_1,R_2)\in \Cr_\mathrm{G}}(R_1+\lambda R_2) \\
&\qquad \stackrel{(a)}{\ge} p_1 \log\left(\frac{a^2(1-\rho^2)+N_1}{N_1}\right) + \bar{p}_1 \log\left(\frac{a^2(1-\rho^2)+N_2}{N_2}\right)+ \lambda p_2 \log\left(\frac{T+N_1}{a^2(1-\rho^2)+N_1}\right) + \lambda\bar{p}_2 \log\left(\frac{T+N_2}{a^2(1-\rho^2)+N_2}\right)\\
&\qquad\stackrel{(b)}{\ge} p_1\log\left(\frac{T+N_1}{b^2(1-\rho^2)+N_1}\right)+\bar{p}_1\log\left(\frac{T+N_2}{b^2(1-\rho^2)+N_2}\right)-\log\left(\frac{1}{1-\rho^2}\right)\\
&\qquad\qquad +\lambda p_2\log\left(\frac{T+N_1}{a^2(1-\rho^2)+N_1}\right)+\lambda \bar{p}_2\log\left(\frac{T+N_2}{a^2(1-\rho^2)+N_2}\right).
\end{align*}

Step $(a)$ follows by plugging in $\alpha = a^2(1-\rho^2)/T$ in~\eqref{cg}. 
To show step $(b)$ note that the difference between the LHS and RHS is
\begin{align}\label{lhs}
&p_1 \log \frac{(a^2(1-\rho^2)+N_1)(b^2(1-\rho^2)+N_1)}{N_1(1-\rho^2)(a^2+b^2+2ab\rho+N_1)} + \bar{p}_1 \log \frac{(a^2(1-\rho^2)+N_2)(b^2(1-\rho^2)+N_2)}{N_2(1-\rho^2)(a^2+b^2+2ab\rho+N_2)} \ge 0
\end{align}
because $(a^2(1-\rho^2)+N_j)(b^2(1-\rho^2)+N_j) - N_j(1-\rho^2)(a^2+b^2+2ab\rho+N_j) = (ab(1-\rho^2) - \rho N_j)^2 \ge 0$ for $j=1,2$. Equality holds for~\eqref{lhs} if and only if $(\rho,a) = (0,0)$ or $(\rho,b)=(0,0)$. Thus equality in step $(b)$ holds if and only if $\max_{(R_1,R_2)\in \Cr_\mathrm{G}}(R_1+\lambda R_2) = \max(C_1, \l C_2)$, which is in general not true. If $\max_{(R_1,R_2)\in \Cr_\mathrm{G}}(R_1+\lambda R_2) \neq \max(C_1, \l C_2)$, it follows that
\begin{align*}
\max_{(R_1,R_2)\in \Cr_\mathrm{G}}(R_1+\lambda R_2) &> \max_{(R_1,R_2)\in \Rr_\mathrm{D}}(R_1+\lambda R_2).
\end{align*}

\end{document}